\documentclass[10pt,twocolumn,letterpaper]{article}
\PassOptionsToPackage{table}{xcolor}

\usepackage{cvpr}              % To produce the CAMERA-READY version
\usepackage{multirow} 
\usepackage{arydshln} % 支持 \cdashline 和 \hdashline
\usepackage{amsmath,amssymb,amsthm}
\usepackage{mathtools}   % refinements to amsmath (e.g., \coloneqq, \DeclarePairedDelimiter)
\usepackage{bm}
\usepackage{pifont}
\usepackage{subcaption}

\usepackage{booktabs}
\usepackage{caption}

\usepackage{array,booktabs,xcolor}

\definecolor{OursRow}{HTML}{E6F0FF}

\newtheorem{thm}{Theorem}[section]
\newtheorem{mydef}[thm]{Definition}

\newtheorem{myrem}{Remark}
\newtheorem{myprop}[thm]{Proposition}

\newcolumntype{V}{!{\vrule width 1.pt}}  % 整列粗竖线

\definecolor{cvprblue}{rgb}{0.21,0.49,0.74}
\usepackage[pagebackref,breaklinks,colorlinks,allcolors=cvprblue]{hyperref}

\def\paperID{20852} % *** Enter the Paper ID here
\def\confName{CVPR}
\def\confYear{2026}

\title{Confusion-Aware Spectral Regularizer for Long-Tailed Recognition}

\author{
Ziquan Zhu$^{1,*}$ \quad
Gaojie Jin$^{2,*}$ \quad
Hanruo Zhu$^{1,*}$ \quad
Si-Yuan Lu$^{3,*}$ \quad
Yunxiao Zhang$^{2}$\\
Zeyu Fu$^{2}$ \quad
Ronghui Mu$^{2}$ \quad
Guoqiang Zhang$^{2}$ \quad
Zhao Sun$^{4}$ \quad
Yuhang Xia$^{5}$\\
Jiaxing Shang$^{2,6}$ \quad
Xiang Li$^{7}$ \quad
Lu Liu$^{2}$ \quad
Tianjin Huang$^{2,8,\dagger}$\\[0.4em]
$^{1}$University of Leicester \quad
$^{2}$University of Exeter \quad
$^{3}$Nanjing University of Posts and Telecommunications \\
$^{4}$Zhengzhou University \quad 
$^{5}$Chengdu University of Technology \quad
$^{6}$Chongqing University \\
$^{7}$University of Bristol \quad
$^{8}$Eindhoven University of Technology\\
{$^{*}$Equal contribution \hspace{1em} $^{\dagger}$Corresponding author: T.Huang2@exeter.ac.uk.}
}

\begin{document}
\maketitle
\begin{abstract}
Long-tailed image classification remains a long-standing challenge, as real-world data typically follow highly imbalanced distributions where a few head classes dominate and many tail classes contain only limited samples. This imbalance biases feature learning toward head categories and leads to significant degradation on rare classes. Although recent studies have proposed re-sampling, re-weighting, and decoupled learning strategies, the improvement on the most underrepresented classes still remains marginal compared with overall accuracy.
In this work, we present a confusion-centric perspective for long-tailed recognition that explicitly focuses on worst-class generalization. 
We first establish a new theoretical framework of class-specific error analysis, which shows that the worst-class error can be tightly upper-bounded by the spectral norm of the frequency-weighted confusion matrix and a model-dependent complexity term. Guided by this insight, we propose the \textbf{C}onfusion-\textbf{A}ware Spectral \textbf{R}egularizer (\texttt{CAR}) that minimizes the spectral norm of the confusion matrix during training to reduce inter-class confusion and enhance tail-class generalization. To enable stable and efficient optimization, \texttt{CAR} integrates a Differentiable Confusion Matrix Surrogate and an EMA-based Confusion Estimator to maintain smooth and low-variance estimates across mini-batches. 
Extensive experiments across multiple long-tailed benchmarks demonstrates that \texttt{CAR} substantially improves both worst-class accuracy and overall performance. When combined with ConCutMix augmentation, \texttt{CAR} consistently surpasses exisiting state-of-the-art long-tailed learning methods under both the training-from-scratch setting (by $2.37\% \sim 4.83\%$) and the fine-tuning-from-pretrained setting (by $2.42\% \sim 4.17\%$) across ImageNet-LT, CIFAR100-LT, and iNaturalist datasets.
Code is available at \url{https://github.com/misswayguy/CAR}.

\end{abstract}
    
\section{Introduction}
\label{sec:intro}

Image classification has achieved remarkable progress in recent years, mainly due to the success of deep learning and large-scale balanced datasets~\cite{rawat2017deep}. However, data collected from real-world scenarios rarely follow such balanced distributions. Instead, they usually exhibit a long-tailed distribution, where a few head classes contain abundant samples while most tail classes have only limited instances~\cite{zhang2023deep,zhang2025systematic}. This severe imbalance causes biased feature learning—models tend to focus on head classes while failing to capture meaningful representations for tail ones~\cite{feldman2020does,hou2023subclass}. 
This skewed distribution induces optimization bias toward head classes and degrades performance on tail categories, making long-tailed recognition a central challenge in practice. To mitigate the negative effects of long-tailed distributions, numerous approaches have been explored in recent years. Data-level strategies such as re-sampling~\cite{kang2019decoupling,wang2020devil,buda2018systematic} aim to balance the number of samples per class by over- or under-sampling the training data. Loss-level adjustments, including re-weighting~\cite{cui2019classbalanced,lin2017focal,ren2020balanced} and logit adjustment~\cite{menon2020long,tian2020posterior,wu2021adversarial}, attempt to compensate for class imbalance by assigning adaptive weights or biases to the loss function. At the model level, representation enhancement and decoupled learning methods~\cite{dong2022lpt,shi2023long} focus on separating feature learning from classifier optimization to reduce the dominance of head classes.

Despite steady progress in long-tailed learning, performance on the worst-performing classes remains substantially inferior to overall accuracy. As shown in Figure~\ref{fig:worst_overall_bars}, we observe two critical gaps: \ding{182} worst-class test accuracy lags significantly behind overall test accuracy, and \ding{183} worst-class test accuracy falls substantially short of worst-class training accuracy.  These disparities reveal a fundamental limitation that existing approaches fail to effectively generalize to the most challenging tail categories, even when they fit the training data well.

\begin{figure*}[htb]
  \centering
  \captionsetup{aboveskip=2pt, belowskip=0pt} % 仅对本图生效
  \includegraphics[width=0.85\linewidth]{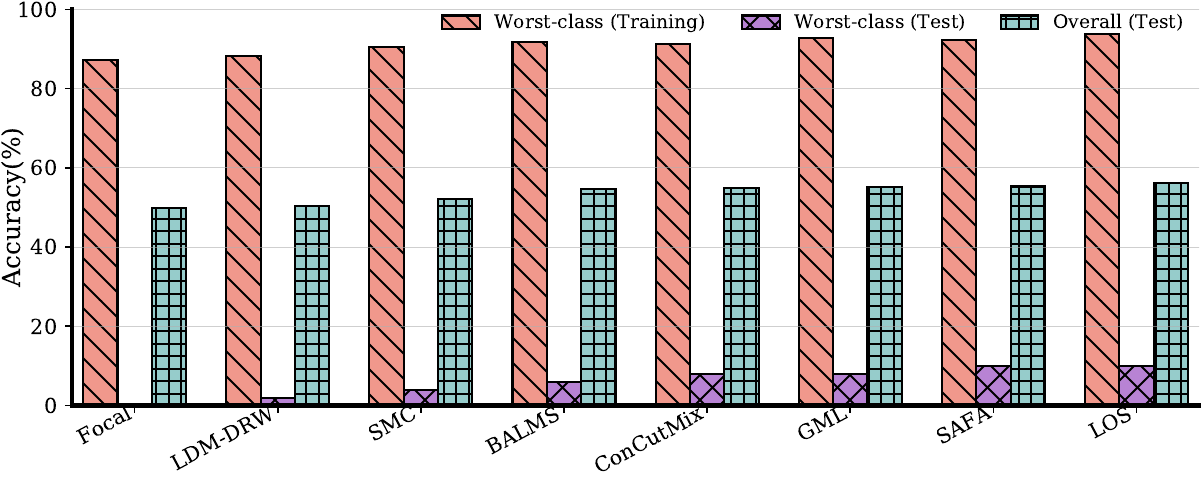}
  \vspace{-4pt}
  \caption{Poor generalization of worst-class performance in existing long-tailed learning methods. Experiments are conducted on ImageNet-LT using ViT-Small as the backbone. The three bars for each method correspond to the worst-class accuracy on the training set (left), the worst-class accuracy on the test set (middle), and the overall test accuracy (right).} 
  \label{fig:worst_overall_bars}
  \vspace{-8pt}
\end{figure*}

To bridge this gap, we propose a new confusion-centric perspective that explicitly regularizes spectral norm of the weighted confusion matrix and focuses on improving worst-class performance. Specifically, we introduce a weighted worst-class error metric that integrates frequency priors to amplify the influence of minority classes. We further develop a generalization upper bound for the class-specific error based on the PAC-Bayesian framework~\cite{mcallester1999pac,morvant2012pac,neyshabur2017pac}. This generalization bound reveals that the worst-class error can be bounded by two key components: (i) the spectral norm of the weighted empirical confusion matrix, and (ii) a model- and data-dependent complexity term. Building on this theoretical insight, we propose a confusion-aware spectral regularizer (\texttt{CAR}) that directly minimizes the spectral norm of the frequency-weighted confusion matrix during training. To enable efficient optimization, we introduce a differentiable confusion matrix surrogate combined with an exponential moving average (EMA) mechanism that maintains stable and efficient estimates across mini-batches.

\noindent Our efforts are unfolded with the following four
thrusts:
\begin{itemize}[leftmargin=*]
    \item[$\star$] \textbf{(Theoretical Analysis)} We establish a novel confusion-centric perspective for long-tailed recognition. Building on the PAC-Bayesian theory, we derive a new upper bound showing that the \emph{worst-class error} can be tightly controlled by the spectral norm of the frequency-weighted confusion matrix, offering a principled route toward improving worst-class generalization. 
    
    \item[$\star$]\textbf{(Algorithm)} Guided by the theoretical insights, we propose the practical Confusion-Aware Spectral Regularizer (\texttt{CAR}). \texttt{CAR} introduces two key components: \textbf{(i)} a \emph{Differentiable Confusion Matrix Surrogate} that replaces non-differentiable indicators with smooth approximations, and \textbf{(ii)} an \emph{EMA-based Confusion Estimator} that stabilizes optimization by maintaining low-variance estimates of confusion matrix. 

    \item[$\star$] \textbf{(Experiments)} We conduct extensive experiments across diverse long-tailed benchmarks, architectures, and imbalance factors. Results consistently demonstrate that \texttt{CAR} achieves superior head–tail balance, and improves both \textit{worst-class} accuracy and \textit{overall} performance. For example, when combined with ConCutMix augmentation, \texttt{CAR} surpasses the previous state-of-the-art LOS~\cite{sunrethinking} by $2.37\% \sim 4.83\%$  under the training-from-scratch setting across CIFAR100-LT, ImageNet-LT, and iNaturalist.

    \item[$\star$] \textbf{(Extra Findings)} We further observe that \texttt{CAR} can be seamlessly integrated with existing data-augmentation-based long-tailed learning methods, leading to additional performance gains on both head and tail classes.
\end{itemize}

% -----------------------------

\section{Preliminaries}

\paragraph{Notation.}
Let $\mathcal{X}\subset\mathbb{R}^d$ denote the input space and $\mathcal{Y}=\{1,\dots,K\}$ the label set. 
A training set $\mathcal{S}=\{(x_q,y_q)\}_{q=1}^{m}$ is drawn from the imbalanced distribution $\mathcal{D}$. 
$m$ is the number of training samples and $K$ is the number of classes. 
A classifier $f:\mathcal{X}\to\mathbb{R}^{K}$ outputs predicted probability vector $f(x)$, and the prediction is $\hat{y}(x)=\arg\max_{i} f(x)[i]$. 
For a matrix $\mathbf{A}=(a_{ij})$, we use $\|\mathbf{A}\|_{1}=\max_{j}\sum_{i}|a_{ij}|$ and $\|\mathbf{A}\|_{2}$ for the spectral norm.

\paragraph{Confusion Matrice.}
The population off-diagonal confusion matrix $\mathbf{C}_\mathcal{D}^f\in\mathbb{R}^{K\times K}$ is defined as
\begin{equation}
c_{ij}=\mathbb{P}_{(x,y)\sim\mathcal{D}}\!\left(\hat{y}(x)=i \mid y=j\right),\qquad c_{jj}=0.
\end{equation}
The column sum $\sum_{i}c_{ij}$ corresponds to the class-conditional error for class $j$. 
Its empirical analogue $\mathbf{C}_S^f$ replaces probabilities with sample averages:
\begin{equation}
\hat c_{ij}=\frac{1}{m_j}\sum_{q:y_q=j} \mathbf{1}\!\left(\hat{y}(x_q)=i\right),\qquad \hat c_{jj}=0,
\end{equation}
where $m_j$ is the number of samples with label $j$ in the training set $\mathcal{S}$. In previous works such as~\cite{farnia2018generalizable,neyshabur2017pac}, the PAC-Bayesian generalization analysis for a DNN is conducted on the margin loss. Following the margin setting, we consider any positive margin $\gamma$ in this work and define the empirical margin confusion matrix $\mathbf{C}_{S,\gamma}^f=(\hat{c}_{ij}^{\gamma})$ as:
{\small
\begin{equation}
\hat c^{\,\gamma}_{ij}=
\begin{cases}
0, \qquad\qquad\qquad\qquad \qquad\qquad\qquad\qquad \qquad\;\;\;\;\; i=j, \\[2pt]
\displaystyle
\sum_{q:y_q=j} \frac{1}{m_j}\,
\mathbf{1}\!\big(f_w(\mathbf{x}_q)[y_q] \le \gamma + f_w(\mathbf{x}_q)[i]\big)
\\[-2pt]\qquad\qquad\qquad\qquad
\times\,
\mathbf{1}\!\big(\underset{i'\ne y_q}{\arg\max}\, f_w(\mathbf{x}_q)[i'] = i\big),\notag
 \;\;\text{else.}
\end{cases}
\label{eq:margin-conf} 
\end{equation}
}

\noindent
where $\mathbf{1}[a \le b] = 1$ if $a<b$, else $\mathbf{1}[a \le b] = 0$.
% -----------------------------

\section{Weighted Worst-Class Error}

\textbf{Rationale.} Long-tail datasets exhibit severe imbalance.  Frequent classes dominate risk minimization while rare classes are poorly controlled.  To counteract this, we introduce a frequency-dependent weighting that amplifies the influence of minority classes in the confusion-based analysis.

\noindent We define the class-wise weight as \( \lambda_j = (m_j + r_0)^{-1/2} \), 
where \( r_0 > 0 \) is a smoothing factor and \( m_j \) denotes the relative frequency of class $j$, i.e., the ratio between the number of samples in class $j$ and the total number of samples in the dataset.
The diagonal weighting matrix is then given by \( \mathbf{\Lambda} = \mathrm{diag}(\lambda_1, \dots, \lambda_K) \).

\begin{mydef}[weighted worst-class error]
\begin{equation}
\mathtt{WCE}(f)=\big\|\mathbf{C}_\mathcal{D}^f \mathbf{\Lambda}\big\|_{1}
=\max_{j}\; \lambda_j \sum_{i} c_{ij}.
\end{equation}
\end{mydef}

\noindent This criterion emphasizes classes with fewer samples by scaling their conditional errors with $\lambda_j$. 

% In practice, we adopt a differentiable surrogate
% \begin{equation}
% \widetilde{\mathrm{WCE}}_{\tau,p}(f)=\big\|\widetilde{C}_{S,\tau}^f W\big\|_{1},
% \end{equation}
% where $\widetilde{C}_{S,\tau}^f$ denotes a smoothed estimator (defined in Sec.~\ref{sec:method}).  
% The spectral norm ($p=2$) is especially suitable since $\|A\|_{1}\le \sqrt{K}\|A\|_{2}$, yielding a smooth alternative to the sharp $\ell_1$ form.

\begin{myprop}[\textbf{Upper Bound}]
\label{prop:weighted-bound}
Consider a training set $\mathcal{S}$ with $m$ samples drawn from a distribution $\mathcal{D}$ over $\mathcal{X} \times \mathcal{Y}$. 
Let $B$ represent the largest $\ell_2$ norm of input samples.
For any $B,n,h>0$, let the base classifier $f_{w}:\mathcal{X}\to\mathcal{Y}$ be an $n$-layer feedforward network with $h$ units per layer and ReLU activation.
%and $w_{min}$ be the class-wise weight. 
Then, for any $\delta\in(0,1)$ and margin $\gamma>0$, with probability at least $1-\delta$ over $S\sim\mathcal{D}^{m}$, we have upper bound for class-specific error $e_j$:
\small
\begin{equation}
\begin{aligned}
e_j &\leq \frac{1}{\lambda_{j}}\bigl\| \mathbf{C}_{\mathcal D}^{f} \mathbf{\Lambda} \bigr\|_{1}
\;\\
&\le\;
\frac{\nu}{\lambda_{j}}\,\underbrace{\bigl\| \mathbf{C}_{\mathcal{S},\gamma}^{f} \mathbf{\Lambda} \bigr\|_{2}}_{\textbf{Empirical spectral norm}}
\;+\;
\underbrace{\mathcal{E}(f,\mathcal{S},\gamma,\delta)}_{\substack{\textbf{Model and training}\\\textbf{set dependence}}}, \forall j,
\end{aligned}
\end{equation}
where $\nu$ is a positve constant which depend on K. $\mathcal{E}(f,\mathcal{S},\gamma,\delta)$ is
\small
\begin{equation}
\mathcal{O}\left(\sqrt{\frac{K}{(m_{min} - 8K)\gamma^2} \left[ \Psi(f_w) + \ln \left( \frac{n m_{min}}{\delta} \right) \right]}\right),
\end{equation}
where $K$ is the number of classes, $m_{min}$ represents the minimal number of examples from $\mathcal{S}$ which belong to the same class, $\Psi(f_w)=B^2n^2h\ln(nh)\prod_{l=1}^n \|W_l\|_2^2 \sum_{l=1}^n \frac{\|W_l\|^2_F}{\|W_l\|^2_2}$, $W_l$ represents the $l$-th weight matrix.
\end{myprop}

\noindent
\emph{Proof.} See Appendix \ref{appendix_proof}. 
\hfill $\square$

\begin{myrem}
The above result demonstrates that \textbf{(i)} the class-specific error is upper-bounded by the weighted worst-class error, \textbf{(ii)} this upper bound is tighter for classes with higher errors and looser for those with lower errors, and is exactly tight for the worst-class error; and \textbf{(iii)} the weighted worst-class error itself can be further bounded by two terms: \ding{182} the spectral norm of the weighted empirical confusion matrix, and \ding{183} a model- and data-dependent complexity term $\Psi(f_{w})$. 
While the second term has been extensively studied through techniques such as spectral normalization of weight matrices~\citep{yoshida2017spectral,farnia2018generalizable}, our contribution is to highlight the role of the spectral norm of the confusion-matrix. Controlling this spectral norm offers a complementary mechanism for improving worst-class generalization, particularly in imbalanced settings.
\end{myrem}

\section{\texttt{CAR}:Confusion-Aware Spectral Regularizer}
\label{sec:method}
The above analysis indicates that reducing the upper bound of the weighted worst-class error effectively decreases the error across all classes, particularly for the worst-performing class. Motivated by this, we introduce the empirical spectral norm term as a regularizer, while omitting the leading constant factor $\tfrac{\nu}{\lambda_j}$ for simplicity. Specifically, it is formulated as follows:
\begin{equation}
\mathcal{R}(f)=
\|\mathbf{C}_{\mathcal{S},\gamma}^f \mathbf{\Lambda}\|_{2}
\end{equation}
\noindent However, optimizing the model with respect to this regularizer presents two key challenges: $(1)$ the empirical margin confusion matrix is non-differentiable; and  
$(2)$ computing it over the entire training set \(\mathcal{S}\) in real time is computationally infeasible.  
To overcome these challenges, we introduce a \texttt{Differentiable Confusion Matrix Surrogate} and an \texttt{EMA-based Confusion Estimator}.
%where $\mathtt{CE}(\cdot)$ denotes the loss of cross entropy and $\lambda>0$ balances the standard cross entropy with the confusion aware regularizer. 

\noindent\textbf{Differentiable Confusion Matrix Surrogate.} The empirical confusion matrix ${\mathbf{C}}_{\mathcal{S},\gamma}^f$ involves non-differentiable indicator functions, making it unsuitable for gradient-based optimization. To address this, we introduce a differentiable surrogate formulation ${\tilde{\mathbf{C}}}_{\mathcal{S},\gamma}^f=(\tilde c_{ij})$:
% \begin{equation}
% \tilde c_{ij}=\frac{1}{m_j}\sum_{q:\,y_q=j}
% \mathrm{softmax}\!\left(f(x_q)\right)[i],
% \qquad \tilde c_{jj}=0.
% \end{equation}
\begin{align}
\label{eq:soft-margin-confusion}
\tilde c_{ij}
=\frac{1}{m_j}\sum_{q:y_q=j}
\underbrace{\sigma (\gamma+f_w(x_q)[i]-f_w(x_q)[j])}_{\textbf{soft margin gate}} \notag
\\\times\, \underbrace{\mathbf{S} ( f_w(x_q)-f_w(x_q)[j])[i]}_{\textbf{soft argmax over non-j}}
\end{align}
where $\sigma(\cdot)$ denotes the sigmoid function and $\mathbf{S}(\cdot)$ the softmax function.
The \textbf{soft margin gate} term softly evaluates whether class $i$ surpasses the ground-truth class $j$ by a margin $\gamma$, while the \textbf{soft argmax over non-\emph{j}} term provides a differentiable approximation of the most competitive class other than the true one.

\noindent\textbf{EMA-based Confusion Estimator.} Considering that training models is conducted with batchsize, a directly way to avoid the infeasible trainset level confusion matrix calcuation is to calcuate batchsize level confusion matrix. However, a single mini-batch may yields a high-variance estimate of the confusion matrix. To stablilize training without recomputing overal full trainset,  we maintain an \emph{exponential moving average} of the \emph{batch-level differentiable estimate}. Specficially, let $\mathcal{B}_t$ denote the mini-batch at iteration $t$ and $\beta\in[0,1)$ be a momentum parameter , the EMA of the batch-level confusion matrix is defined as:
\begin{equation}
\label{eq:ema-conf}
 \mathbf{\hat C_t}
\;=\;
\beta\,( \mathbf{\hat C}_{t-1})
\;+\;
(1-\beta)\, \mathbf{\tilde C}_{\mathcal{B}_t,\gamma}^f,
\qquad
\mathbf{\hat C_0} = 0.
\end{equation}

\noindent Only the current term $ \mathbf{\tilde C}_{\mathcal{B}_t,\gamma}^f$ carries gradients w.r.t.\ $w$; the history $ \mathbf{\hat C}_{t-1}$ is treated as a constant, preserving differentiability while reducing variance (See Appendix~\ref{appedix:stable} for a detailed stability analysis.)

\noindent The overall training objective $\mathcal{L}(f)$ combines the standard cross-entropy loss with the proposed confusion-aware regularization term, where $\alpha>0$ controls the strength of regularization.
\begin{equation}
\mathcal{L}(f)=
\frac{1}{m}\sum_{q=1}^{m}\mathtt{CE}\!\big(f(x_q),y_q\big)
+\alpha\,\|\!\big(\mathbf{\hat C_t} \mathbf{\Lambda}\big)\|_{2},
\end{equation}
where $\mathtt{CE}(\cdot)$ denote the cross-entropy loss.

\begin{table*}[tbh]
    \centering
    \caption{Top-1 accuracy (\%) comparison on ImageNet-LT, CIFAR100-LT (IF=100), and iNaturalist using ViT-Small.
    Results are reported for Head, Medium, Tail, and Overall.
    The best results are in \textbf{bold}, and the second-best are \underline{underlined}.}
    \setlength{\tabcolsep}{3.5pt}
    \renewcommand{\arraystretch}{0.95}
    \resizebox{0.99\textwidth}{!}{
    % 关键：用 V 放置“整列粗竖线”；表体内部无其它竖线
    \begin{tabular}{l r V cccc V cccc V cccc}
        \toprule
        \textbf{Methods} & \textbf{Venue} &
        \multicolumn{4}{c}{\textbf{ImageNet-LT}} &
        \multicolumn{4}{c}{\textbf{CIFAR100-LT}} &
        \multicolumn{4}{c}{\textbf{iNaturalist}} \\
        % 只画分块的横线；不对 Venue 画横线；不画 Head/Overall 的小竖线
        \cmidrule(lr){3-6}\cmidrule(lr){7-10}\cmidrule(lr){11-14}
        & & Head & Medium & Tail & Overall & Head & Medium & Tail & Overall & Head & Medium & Tail & Overall \\
        \midrule
        \rowcolor{gray!15} CE & -- &
        69.71 & 43.91 & 16.30 & 46.51 &
        68.20 & 45.37 & 15.17 & 41.40 &
        68.66 & 63.54 & 58.82 & 59.56 \\
        Focal~\cite{lin2017focal} & ICCV’2017 &
        67.57 & 49.80 & 22.83 & 49.78 &
        66.43 & 47.49 & 21.87 & 43.13 &
        67.74 & 68.55 & 64.96 & 66.86 \\
        CB~\cite{cui2019classbalanced} & CVPR’2019 &
        69.69 & 47.69 & 19.90 & 48.05 &
        67.63 & 45.77 & 16.37 & 42.10 &
        67.69 & 66.66 & 62.39 & 64.71 \\
        LDAM-DRW~\cite{cao2019learning} & NeurIPS’2019 &
        68.06 & 47.14 & 25.90 & 50.39 &
        68.97 & 46.40 & 25.07 & 45.40 &
        67.15 & 68.30 & 63.64 & 65.57 \\
        BALMS~\cite{ren2020balms} & NeurIPS’2020 &
        69.80 & 53.14 & 30.23 & 54.61 &
        68.80 & 58.57 & 28.03 & 50.74 &
        70.48 & 70.19 & 69.53 & 70.55 \\
        ReMix~\cite{chou2020remix} & ECCV’2020 &
        69.49 & 46.94 & 25.67 & 50.05 &
        69.06 & 46.97 & 20.10 & 43.34 &
        68.14 & 66.49 & 62.55 & 64.27 \\
        BBN~\cite{zhou2020bbn} & CVPR’2020 &
        69.11 & 50.26 & 29.37 & 52.79 &
        \textbf{69.21} & 57.71 & 29.42 & 50.85 &
        68.58 & 69.29 & 65.72 & 67.44 \\
        MetaSAug~\cite{li2021metasaug} & CVPR’2021 &
        66.63 & 47.86 & 30.57 & 51.94 &
        67.31 & 47.91 & 27.57 & 45.55 &
        \underline{69.15} & 68.14 & 67.82 & 68.00 \\
        CMO~\cite{park2022majority} & CVPR’2022 &
        67.40 & 50.20 & 28.97 & 52.15 &
        68.34 & 49.37 & 28.53 & 46.01 &
        68.57 & 70.31 & 69.60 & 69.41 \\
        SAFA~\cite{hong2022safa} & ECCV’2022 &
        67.20 & 52.03 & 32.03 & 55.29 &
        67.29 & 54.94 & 29.47 & 48.02 &
        68.86 & 72.03 & 70.42 & 70.86 \\
        WB~\cite{alshammari2022long} & CVPR’2022 &
        \underline{70.46} & 49.80 & 31.90 & 53.71 &
        68.86 & 49.94 & 29.60 & 47.81 &
        70.10 & 69.74 & 67.95 & 69.35 \\
        GML~\cite{du2023no} & CVPR’2023 &
        69.03 & 53.66 & 32.17 & 55.24 &
        68.63 & 55.20 & 30.97 & 50.23 &
        \textbf{70.71} & 70.62 & 69.38 & 70.85 \\
        ConCutMix~\cite{pan2024enhanced} & TIP’2024 &
        69.23 & 48.94 & 31.00 & 54.97 &
        67.83 & 53.54 & 30.60 & 47.86 &
        68.64 & 71.48 & 70.32 & 70.24 \\
        LOS~\cite{sunrethinking} & ICLR’2025 &
        \textbf{70.79} & 55.31 & 32.73 & 56.20 &
        \textbf{69.21} & 57.71 & 29.42 & 50.85 &
        68.50 & \underline{71.43} & 72.14 & 71.01 \\
        \rowcolor{red!10} \texttt{CAR} (Ours) & -- &
        67.00 & \underline{54.57} & \underline{35.77} & \underline{57.48} &
        65.37 & \underline{58.60} & \underline{34.03} & \underline{51.85} &
        68.52 & 71.32 & \underline{73.73} & \underline{71.56} \\
        \rowcolor{red!10}\texttt{CAR} (Ours) + ConCutMix & -- &
        69.56 & \textbf{55.94} & \textbf{38.07} & \textbf{60.07} &
        \underline{68.89} & \textbf{60.14} & \textbf{37.40} & \textbf{55.68} &
        69.10 & \textbf{72.05} & \textbf{75.42} & \textbf{73.38} \\
        \bottomrule
    \end{tabular}
    }
    \label{tab:lt_main}
\end{table*}

\section{Experiments}
\label{sec:experiments}

\subsection{Experimental Setup}

\textbf{Datasets.}~We conduct extensive experiments on four widely used long-tailed benchmarks: CIFAR100-LT~\cite{krizhevsky2009learning}, ImageNet-LT~\cite{liu2019large}, Tiny-ImageNet-LT~\cite{le2015tiny}, and iNaturalist2018~\cite{van2018inaturalist}. 
Following~\cite{li2025focal}, we adopt the same data construction and evaluation protocols. 
For CIFAR-100-LT, ImageNet-LT, and Tiny-ImageNet-LT, the imbalanced training sets are generated by sampling the original balanced datasets according to an exponential distribution, with imbalance factors of 200, 100, and 50, respectively. 
The test sets remain class-balanced. 
The iNaturalist2018 dataset contains 437.5K natural images from 8,142 categories, exhibiting a naturally long-tailed distribution without synthetic resampling. 
Consistent with~\cite{zhao2024ltrl}, we report results on three subsets according to the number of training samples per class: 
Head (more than 100 images), Medium (20$\sim$100 images), and Tail (less than 20 images).
More details abour datasets would be discussed in the Appendix~\ref{appedix:datasets}.

\begin{table*}[tbh]
    \centering
    \caption{Top-1 accuracy (\%) comparison on Tiny-ImageNet-LT, CIFAR100-LT (IF=100), and iNaturalist using pre-trained ViT-Small.
    Results are reported for Head, Medium, Tail, and Overall.
    The best results are in \textbf{bold}, and the second-best are \underline{underlined}.}
    \setlength{\tabcolsep}{3.5pt}
    \renewcommand{\arraystretch}{0.95}
    \resizebox{0.99\textwidth}{!}{
    \begin{tabular}{l r V cccc V cccc V cccc}
        \toprule
        \textbf{Methods} & \textbf{Venue} &
        \multicolumn{4}{c}{\textbf{Tiny-ImageNet-LT}} &
        \multicolumn{4}{c}{\textbf{CIFAR100-LT}} &
        \multicolumn{4}{c}{\textbf{iNaturalist}} \\
        \cmidrule(lr){3-6}\cmidrule(lr){7-10}\cmidrule(lr){11-14}
        & & Head & Medium & Tail & Overall & Head & Medium & Tail & Overall & Head & Medium & Tail & Overall \\
        \midrule
        \rowcolor{gray!15} CE & -- &
        87.31 & 71.66 & 39.90 & 63.91 &
        90.46 & 73.37 & 50.50 & 70.49 &
        75.47 & 73.95 & 71.82 & 73.45 \\
        Focal~\cite{lin2017focal} & ICCV’2017 &
        88.57 & 74.11 & 44.23 & 67.31 &
        91.69 & 77.66 & 54.63 & 73.66 &
        75.15 & 78.21 & 78.93 & 79.59 \\
        CB~\cite{cui2019classbalanced} & CVPR’2019 &
        85.60 & 75.31 & 43.58 & 66.99 &
        91.83 & 76.17 & 51.57 & 72.27 &
        75.36 & 75.84 & 75.93 & 77.08 \\
        LDAM-DRW~\cite{cao2019learning} & NeurIPS’2019 &
        85.91 & 75.29 & 43.90 & 65.28 &
        91.31 & 75.66 & 52.87 & 72.40 &
        75.87 & 77.84 & 77.86 & 79.05 \\
        BALMS~\cite{ren2020balms} & NeurIPS’2020 &
        \underline{89.54} & \underline{76.69} & 46.77 & 70.11 &
        \textbf{93.80} & 78.37 & 56.93 & 77.44 &
        \textbf{77.60} & 80.24 & 82.56 & 82.34 \\
        ReMix~\cite{chou2020remix} & ECCV’2020 &
        85.49 & 75.51 & 43.73 & 65.19 &
        92.09 & 73.63 & 53.37 & 72.56 &
        74.57 & 77.59 & 76.51 & 78.22 \\
        BBN~\cite{zhou2020bbn} & CVPR’2020 &
        86.27 & 75.78 & 44.94 & 66.47 &
        92.23 & 76.00 & 53.37 & 74.89 &
        75.07 & 79.57 & 79.83 & 80.88 \\
        MetaSAug~\cite{li2021metasaug} & CVPR’2021 &
        86.26 & 74.40 & 47.40 & 67.40 &
        92.77 & 75.46 & 55.20 & 74.34 &
        75.62 & 78.86 & 79.84 & 80.06 \\
        CMO~\cite{park2022majority} & CVPR’2022 &
        85.94 & 75.97 & 45.70 & 67.38 &
        91.57 & 76.14 & 57.07 & 75.37 &
        75.45 & 79.12 & 81.38 & 81.52 \\
        SAFA~\cite{hong2022safa} & ECCV’2022 &
        86.09 & 73.20 & 49.97 & 68.99 &
        92.57 & 77.23 & 59.50 & 77.88 &
        76.39 & 80.09 & 83.07 & 82.77 \\
        WB~\cite{alshammari2022long} & CVPR’2022 &
        88.97 & 74.87 & 45.82 & 67.51 &
        \underline{93.54} & 77.63 & 53.23 & 75.78 &
        76.04 & 80.09 & 79.93 & 81.33 \\
        GML~\cite{du2023no} & CVPR’2023 &
        \textbf{89.87} & 75.93 & 48.87 & 70.43 &
        92.09 & 79.11 & 57.10 & 77.40 &
        \underline{76.95} & 80.69 & 83.10 & 82.79 \\
        ConCutMix~\cite{pan2024enhanced} & TIP’2024 &
        88.23 & 76.54 & 46.27 & 69.15 &
        92.40 & 76.83 & 58.77 & 76.46 &
        75.86 & 79.55 & 83.32 & 82.06 \\
        LOS~\cite{sunrethinking} & ICLR’2025 &
        88.36 & 76.83 & 49.86 & 71.67 &
        93.02 & 80.06 & 58.14 & 78.50 &
        76.43 & 80.58 & 83.18 & 83.02 \\
        \rowcolor{red!10} \texttt{CAR} (Ours) & -- &
        87.60 & 75.91 & \underline{52.47} & \underline{72.97} &
        92.89 & \underline{80.46} & \underline{61.17} & \underline{79.37} &
        75.38 & \underline{81.71} & \underline{84.26} & \underline{83.74} \\
        \rowcolor{red!10} \texttt{CAR} (Ours) + ConCutMix & -- &
        88.71 & \textbf{76.84} & \textbf{54.23} & \textbf{75.84} &
        93.26 & \textbf{81.24} & \textbf{63.33} & \textbf{82.12} &
        76.32 & \textbf{82.08} & \textbf{85.40} & \textbf{85.44} \\
        \bottomrule
    \end{tabular}
    }
    \label{tab:lt_pretrained}
\end{table*}

\noindent \textbf{Implementation Details.}~We adopt ViT-Small~\cite{dosovitskiy2021an} as the main backbone, while additional results on different model sizes (Tiny, Base, and Large) and architectures (ResNet~\cite{he2016deep} and Swin Transformer~\cite{liu2021swin}) are also included to demonstrate model generalization.

For training from scratch, we follow the standard protocol used in long-tailed recognition~\cite{hong2021disentangling}:
models on CIFAR100-LT and  ImageNet-LT are trained for 200 epochs, while iNaturalist2018 uses 300 epochs due to its greater intra-class variability and inherent natural imbalance.
For fine-tuning from pre-trained models, we adopt a shorter schedule of 100 epochs, which is consistent with common practice for adapting large-scale pretrained backbones to long-tailed settings.
All experiments use the AdamW optimizer.
The batch size is fixed to 128 across all datasets.
% Unless otherwise specified, we set $\beta{=}0.3$, $r_{0}{=}0.2$, $\alpha{=}0.5$, and $\gamma{=}0.0$.
Additional training settings are provided in the Appendix~\ref{appedix:implementation_details}.

\noindent \textbf{Baselines.}~We compare our method with a comprehensive set of baselines, covering standard classification, long-tailed data augmentation, and representative long-tailed learning approaches.
($1$)~\textit{Standard Methods.}
To provide a fair comparison, we first include widely used classification baselines: 
Cross-Entropy (CE).
($2$)~\textit{LT Data Augmentation Methods.}
We further compare with a series of augmentation-based long-tailed strategies, including 
ReMix~\cite{chou2020remix}, MetaSAug~\cite{li2021metasaug}, 
CMO~\cite{park2022majority}, ConCutMix~\cite{pan2024enhanced}, and SAFA~\cite{hong2022safa}.
($3$)~\textit{Other Long-Tailed Recognition Methods.}
We evaluate against representative long-tailed learning methods, including 
Class-Balanced Loss (CB)~\cite{cui2019classbalanced}, 
Focal Loss~\cite{lin2017focal}, 
BALMS~\cite{ren2020balms}, 
GML~\cite{du2023no}, 
Weight Balancing (WB)~\cite{alshammari2022long}, 
BBN~\cite{zhou2020bbn}, 
LDAM-DRW~\cite{cao2019learning}, and LOS~\cite{sunrethinking}.

\subsection{Superior Performance on Long-tailed Datasets} \label{full:main}

\textbf{Training from Scratch.}~We evaluate the proposed \texttt{CAR} under a strict training-from-scratch setting. Experiments are conducted on three representative long-tailed benchmarks: CIFAR100-LT~\cite{krizhevsky2009learning}, ImageNet-LT~\cite{liu2019large}, and iNaturalist2018~\cite{van2018inaturalist}, using ViT-Small as the backbone. The results are summarized in Table~\ref{tab:lt_main}. \underline{\emph{We observe}} that \texttt{CAR} consistently achieves the best overall performance across all three long-tailed benchmarks, outperforming all competing methods. \underline{\emph{In addition}}, \texttt{CAR} improves the best tail accuracy by $1.59\% \sim 4.61\%$, demonstrating its strong effectiveness in enhancing tail-class performence. \underline{\emph{Furthermore}}, when combined with the ConCutMix augmentation, \texttt{CAR} surpasses the current state-of-the-art LOS~\cite{sunrethinking} by $2.37\% \sim 4.83\%$ and $3.28\% \sim 7.98\% $ in overall and tail accuracy, respectively, across the three datasets. For example, on ImageNet-LT, our method improves the previous best overall accuracy from $56.20\%$ to $60.07\%$, and boosts the best tail accuracy from $32.73\%$ to $38.07\%$.

\noindent \textbf{Fine-tuning from Pre-trained Models.}~ Our success extends beyond the training-from-scratch setting. We further evaluate \texttt{CAR} and all baselines under the fine-tuning setting, using a pre-trained ViT-Small as the backbone. Results are reported in Table~\ref{tab:lt_pretrained} \emph{Consistent observation can be made across all the three datasets}: 
\ding{182} \texttt{CAR} delivers substantial improvements in both tail accuracy and overall accuracy, highlighting its effectiveness in adapting pre-trained models to long-tailed distributions; 
\ding{183} when combined with ConCutMix, \texttt{CAR} achieves the strongest performance, boosting tail accuracy by $2.22\% \sim 5.19\%$ and overall accuracy by $2.42\% \sim 4.17\%$, significantly outperforming existing state-of-the-art methods.  For example, on iNaturalist, the best overall accuracy improves from $83.02\%$ to $85.44\%$ and the best tail accuracy increases from $83.18\%$ to $85.40\%$.

%As summarized in Table~\ref{tab:lt_pretrained}, our method consistently achieves the best overall accuracies on all three benchmarks and exhibits leading performance in most Medium and Tail partitions. Compared with previous state-of-the-art methods, \texttt{CAR} significantly enhances the adaptability of pre-trained representations to imbalanced data distributions, while maintaining strong generalization on head classes. On ImageNet-LT, our approach achieves the top overall accuracy and clear improvements in tail recognition, indicating that our confusion spectral regularization can effectively refine the feature geometry learned from large-scale pre-training. On CIFAR100-LT and iNaturalist2018, \texttt{CAR} delivers substantial gains across all data regimes, showing remarkable consistency under both synthetic and naturally occurring long-tailed distributions. Furthermore, incorporating ConCutMix~\cite{pan2024enhanced} yields additional gains across all benchmarks, pushing the performance boundaries on ImageNet-LT, CIFAR100-LT, and iNaturalist2018. These results confirm that \texttt{CAR} not only preserves the advantages of pre-training but also serves as a versatile and complementary regularization framework that further amplifies the effectiveness of modern augmentation-based strategies for long-tailed fine-tuning.

\noindent \textbf{Worst-Class Performance.}~To evaluate the effectiveness of \texttt{CAR} in enhancing the generalization of the most underrepresented category, we report the worst-class accuracy on both the training and test sets, along with the Worst-class Ratio (WR = Test/Training), which quantifies the generalization ability of the worst-performing class. Experiments are conducted on ImageNet-LT and CIFAR100-LT using ViT-Small as the backbone.  The results in Table~\ref{tab:wgr_lt} demonstrate that the proposed \texttt{CAR} substantially enhances worst-class generalization. \ding{182} Existing methods exhibit extremely poor worst-class accuracy, typically below $10\%$. \ding{183} In contrast, \texttt{CAR} improves worst-class accuracy by $8\%$ on ImageNet-LT and $6\%$ on CIFAR100-LT, while \texttt{CAR} + ConCutMix further boosts the gains to over $10\%$ on both datasets. 
\ding{184} Moreover, \texttt{CAR} yields a notable increase in the Worst-class Ratio (WR), indicating stronger generalization from training to test time.

\begin{table}[tbh]
\centering
\footnotesize
\setlength{\tabcolsep}{2pt}
\renewcommand{\arraystretch}{0.95}
\caption{Worst-class accuracy on training/test sets and the Worst-class Ratio (WR = Test/Training) based on ViT-Small. All results are presented as percentages.
The best results are highlighted in \textbf{bold}, and the second-best are \underline{underlined}.}
\label{tab:wgr_lt}
\resizebox{\columnwidth}{!}{%
\begin{tabular}{l V ccc V ccc}
\toprule
\multirow{2}{*}{\textbf{Methods}} &
\multicolumn{3}{c}{\textbf{ImageNet-LT}} &
\multicolumn{3}{c}{\textbf{CIFAR100-LT}} \\
\cmidrule(lr){2-4}\cmidrule(lr){5-7}
& Training$(\%)$ & Test$(\%)$ & WR 
& Training$(\%)$ & Test$(\%)$ & WR \\
\midrule
Focal~\cite{lin2017focal}        & 87.22 &  0  & 0.00 & 85.45 &  2  & 0.02 \\
CB~\cite{cui2019classbalanced}   & 84.59 &  0  & 0.00 & 80.00 &  0  & 0.00 \\
BALMS~\cite{ren2020balms}        & 91.81 &  6  & 0.07 & 90.91 &  5  & 0.05 \\
CMO~\cite{park2022majority}      & 90.45 &  4  & 0.04 & 86.00 &  6  & 0.07 \\
SAFA~\cite{hong2022safa}         & 92.35 & 10  & 0.11 & 91.09 &  8  & 0.09 \\
GML~\cite{du2023no}              & 92.83 &  8  & 0.09 & 90.91 &  8  & 0.09 \\
ConCutMix~\cite{pan2024enhanced} & 91.26 &  8  & 0.09 & 87.82 &  8  & 0.09 \\
LOS~\cite{sunrethinking}         &93.72	 & 10  & 0.11 & 91.23 &  8  & 0.09\\

\rowcolor{red!10}
\texttt{CAR} (Ours)                    & \underline{94.24} & \underline{18} & \underline{0.19}
                                 & \underline{92.73} & \underline{14} & \underline{0.15} \\
 \rowcolor{red!10}
\texttt{CAR} (Ours) + ConCutMix       &   \textbf{94.85} &   \textbf{22}  &  \textbf{0.23}  & \textbf{93.17} & \textbf{18}& \textbf{0.19}\\                                
\bottomrule
\end{tabular}
}
\end{table}

\subsection{Generalization Across Backbones} \label{full:backbone}
To verify the effectiveness of our method across different architectures, we further evaluate \texttt{CAR} on five representative backbones, including ViT-Tiny, ViT-Base, ViT-Large, ResNet~\cite{he2016deep} and Swin Transformer~\cite{liu2021swin}, as summarized in Table~\ref{tab:backbone}. 
\texttt{CAR} consistently achieves the highest top-1 accuracy across all evaluated backbones, demonstrating that its effectiveness is not tied to any specific architectural design. 
Notably, on larger transformer backbones, \texttt{CAR} yields consistent gains, including +1.0\% on ViT-Large and +0.8\% on Swin, suggesting that the method integrates effectively with self-attention architectures.
The improvements observed on ResNet further confirm that the proposed confusion-aware spectral regularization also benefits traditional CNN architectures. Overall, the consistent performance gains across diverse model families verify that \texttt{CAR} is a backbone-agnostic regularization method that can be seamlessly incorporated into a wide range of architectures to deliver stable and transferable improvements. Additional results are reported in Appendix~\ref{appedix:models}.

%\texttt{CAR} consistently achieves the highest top-1 accuracy across all backbones, demonstrating that its effectiveness is not limited to a specific model design. Notably, the significant improvements on transformer-based architectures (e.g., +1.6\% on ViT-Large and +1.7\% on Swin) indicate that our method harmonizes well with both self-attention and convolutional feature extractors. The gains on ResNet confirm that the proposed confusion spectral regularization work well in traditional CNNs. Such consistent superiority across model families verifies that \texttt{CAR} serves as a backbone-agnostic and plug-and-play regularization framework, which can be seamlessly integrated into diverse architectures to yield stable and transferable accuracy gains. More detailed backbone-wise results are provided in the Appendix~\ref{appedix:models} and Appendix~\ref{appedix:models_size}.

\subsection{Generalization Across Imbalance Factors} \label{full:if}
We further examine the adaptability of \texttt{CAR} under different imbalance factors (IF = 50 and 200) on ImageNet-LT and CIFAR100-LT using ViT-Small as the backbone. As shown in Table~\ref{tab:lt_ratio}, our method achieves the highest overall accuracy across all settings, demonstrating its strong ability to handle varying levels of class skewness. While the performance of existing methods degrades noticeably as the imbalance increases, \texttt{CAR} maintains a more stable trend, outperforming the second-best approach by clear margins under both moderate (IF=50) and extreme (IF=200) conditions. This consistent behavior indicates that the proposed method effectively mitigates head-class dominance and improves representation alignment for tail categories. Furthermore, the smaller performance gap between IF=50 and IF=200 shows that \texttt{CAR} preserves balanced learning dynamics even when minority classes are extremely rare. More extensive results and analyses across additional imbalance settings can be found in the Appendix~\ref{appedix:if}.

\begin{table}[tb]
\centering
\begingroup
\footnotesize
\setlength{\tabcolsep}{3.6pt}
\renewcommand{\arraystretch}{0.93}
\caption{Top-1 accuracy (\%) on ImageNet-LT across different backbones. 
Results include ViT variants (Tiny/Base/Large), ResNet, and Swin.
The best results are highlighted in \textbf{bold}.}
\label{tab:backbone}
\begin{tabular}{lccccc}
\toprule
\textbf{Methods} & \textbf{ViT-Tiny} & \textbf{ViT-Base} & \textbf{ViT-Large} & \textbf{ResNet} & \textbf{Swin} \\
\midrule
Focal~\cite{lin2017focal}           & 37.98 & 54.44 & 60.66 & 42.31 & 50.92 \\
CB~\cite{cui2019classbalanced}      & 35.09 & 52.65 & 59.14 & 40.06 & 48.79 \\
BALMS~\cite{ren2020balms}           & 45.71 & 62.55 & 67.92 & 48.73 & 55.17 \\
ReMix~\cite{chou2020remix}          & 37.13 & 54.10 & 62.10 & 41.83 & 49.08 \\
MetaSAug~\cite{li2021metasaug}      & 40.80 & 56.70 & 64.22 & 44.47 & 52.52 \\
CMO~\cite{park2022majority}      & 40.07 & 57.36 & 64.71 & 45.67 & 52.60 \\
SAFA~\cite{hong2022safa}            & 43.00 & 60.34 & 67.46 & 46.33 & 54.43 \\
GML~\cite{du2023no}                 & 45.24 & 61.86 & 68.63 & 48.77 & 55.19 \\
ConCutMix~\cite{pan2024enhanced}    & 43.26 & 58.48 & 66.17 & 45.73 & 54.48 \\
LOS~\cite{sunrethinking}            &45.66  & 62.54 & 68.26 & 49.54	& 55.62\\
\rowcolor{red!10}
\texttt{CAR} (Ours)                               & \textbf{46.35} &\textbf{ 63.79} & \textbf{69.26} & \textbf{50.27} & \textbf{56.38} \\
\bottomrule
\end{tabular}
\endgroup
\end{table}

\subsection{Visualization}
To further illustrate the effect of our method on inter-class discrimination, we visualize the class-wise confusion matrices of different approaches on CIFAR100-LT using ViT-Small, as shown in Figure~\ref{fig:confusion_matrix_comparison}. Specifically, we randomly select \textit{ten} categories spanning head, medium, and tail regions to provide a balanced view of class-wise interactions under long-tailed distributions. The left group corresponds to models trained from scratch, while the right group represents fine-tuning from a pre-trained model. Compared with WB and BALMS, Our \texttt{CAR} yields substantially fewer high-intensity off-diagonal responses on both training-from-scratch and fine-tuning-from-pre-trained settings, suggesting that it effectively suppresses inter-class confusion and enhances the separability between categories. 
%This trend holds consistently across all regimes, demonstrating that \texttt{CAR} effectively mitigates inter-class interference and enforces more discriminative decision boundaries. Moreover, in the fine-tuning scenario, the confusion patterns of \texttt{CAR} become even more compact and homogeneous, suggesting that the spectral regularization and pre-trained representations work synergistically to enhance class alignment. 

\begin{table}[tbh]
\centering
\begingroup
\small
\setlength{\tabcolsep}{3.6pt}
\caption{Top-1 accuracy (\%) across different imbalance factors (IF) on ImageNet-LT and CIFAR100-LT based on the ViT-Small.
The best results are highlighted in \textbf{bold}.}
\label{tab:lt_ratio}
\begin{tabular}{lcccc}
\toprule
\multirow{2}{*}{\textbf{Methods}} & \multicolumn{2}{c}{\textbf{ImageNet-LT}} &
\multicolumn{2}{c}{\textbf{CIFAR100-LT}} \\
\cmidrule(lr){2-3}\cmidrule(lr){4-5}
& IF=50 & IF=200 & IF=50 & IF=200 \\
\midrule
CE                 & 49.33 & 35.99 & 44.54 & 37.19 \\
% RS                 & 53.88 & 39.81 & 48.50 & 41.85 \\
% RW                  & 55.38 & 38.70 & 48.44 & 41.40 \\
Focal~\cite{lin2017focal}                & 52.14 & 37.27 & 47.60 & 40.08 \\
CB~\cite{cui2019classbalanced}           & 50.65 & 36.13 & 45.38 & 39.03 \\
LDAM-DRW~\cite{cao2019learning}          & 54.96 & 40.78 & 49.57 & 40.96 \\
BALMS~\cite{ren2020balms}                & 59.57 & 42.18 & 53.34 & 43.85 \\
ReMix~\cite{chou2020remix}               & 53.32 & 37.62 & 48.34 & 39.79 \\
BBN~\cite{zhou2020bbn}                 & 55.26 & 41.89 & 51.79 & 41.43 \\
MetaSAug~\cite{li2021metasaug}           & 55.97 & 38.23 & 50.94 & 42.37 \\
CMO~\cite{park2022majority}          & 57.27 & 42.66 & 51.17 & 43.44 \\
SAFA~\cite{hong2022safa}                 & 58.77 & 43.30 & 53.16 & 43.72 \\
WB~\cite{alshammari2022long}             & 57.93 & 41.76 & 51.67 & 42.68 \\
GML~\cite{du2023no}                      & 59.52 & 43.65 & 54.40 & 45.23 \\
ConCutMix~\cite{pan2024enhanced}         & 58.14 & 43.20 & 50.40 & 44.80 \\
LOS~\cite{sunrethinking}                 & 60.11 & 44.14 & 54.82 & 45.62\\

\rowcolor{red!10}
\texttt{CAR} (Ours)                                     & \textbf{61.10} & \textbf{46.88} & \textbf{55.93} & \textbf{46.30} \\
\bottomrule
\end{tabular}
\endgroup
\end{table}

\begin{figure*}[!thb]
    \centering
    \includegraphics[width=0.80\linewidth]{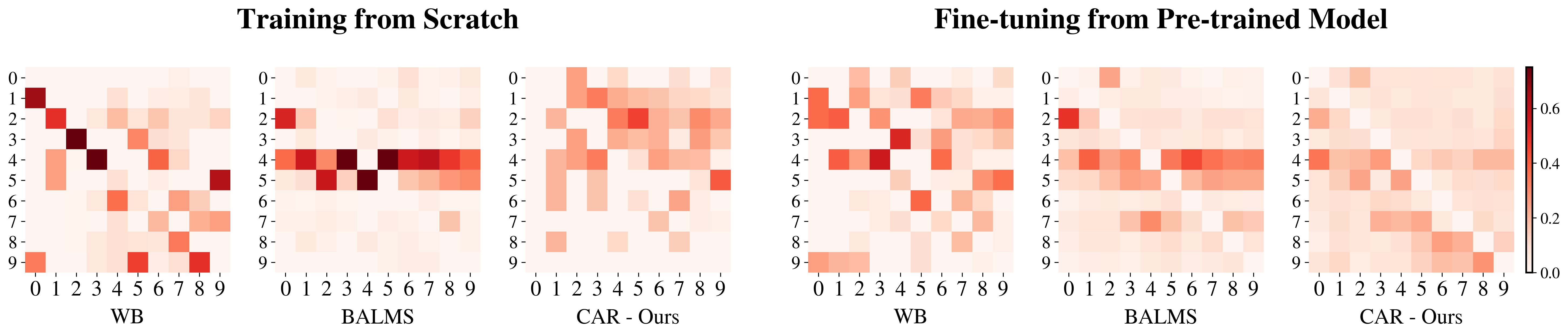}
    \caption{Class-wise confusion matrices on CIFAR100-LT using ViT-Small. Left: training from scratch. Right: fine-tuning from pre-trained model.}
    \label{fig:confusion_matrix_comparison}
\end{figure*}

\subsection{Complementary to Data Augmentation} \label{full:dataaug}
We further examine whether \texttt{CAR} can complement existing long-tailed augmentation strategies as a general regularization module. To this end, we combine \texttt{CAR} with five representative approaches under different imbalance factors on ImageNet-LT and CIFAR100-LT. All models are trained under identical configurations to ensure a fair comparison across varying augmentation paradigms.
As presented in Table~\ref{tab:lt_aug}, \texttt{CAR} consistently boosts the performance of all augmentation methods across datasets and imbalance levels. 
The gains are most pronounced when integrated with ConCutMix and SAFA, achieving the highest overall accuracies under every setting. 
These consistent improvements verify that the proposed spectral regularization yields complementary benefits, enhancing model performance without overlapping effects with existing augmentation techniques.
These results suggest that \texttt{CAR} promotes better class-wise feature separation and facilitates more effective exploitation of augmented samples. 
The observed synergy between model-level regularization and data-level transformations confirms that \texttt{CAR} provides a generalizable mechanism that can be seamlessly integrated with various augmentation frameworks to further enhance long-tailed recognition.
More detailed results and visual comparisons are presented in the Appendix~\ref{appedix:dataaug}.

\begin{table}[!thb]
\centering
\footnotesize
\setlength{\tabcolsep}{2.8pt}
\renewcommand{\arraystretch}{0.96}
\caption{Top-1 accuracy (\%) of ViT-Small on ImageNet-LT and CIFAR100-LT under different imbalance factors (IF).
``+'' indicates the combination of our method with other long-tailed data augmentation methods.}
\label{tab:lt_aug}
\begin{tabular}{l V ccc V ccc}
\toprule
\multirow{2}{*}{\textbf{Methods}} &
\multicolumn{3}{c}{\textbf{ImageNet-LT}} &
\multicolumn{3}{c}{\textbf{CIFAR100-LT}} \\
\cmidrule(lr){2-4}\cmidrule(lr){5-7}
& IF=50 & IF=100 & IF=200 & IF=50 & IF=100 & IF=200 \\
\midrule
\texttt{CAR} (Ours)                       & 61.10 & 57.48 & 46.88 & 55.93 & 51.85 & 46.30 \\
\quad + ReMix~\cite{chou2020remix}       & 62.05 & 58.74 & 48.13 & 56.77 & 53.64 & 48.57 \\
\quad + MetaSAug~\cite{li2021metasaug}   & 64.25 & 59.28 & 49.98 & 57.64 & 53.75 & 47.70 \\
\quad + CMO~\cite{park2022majority}      & 64.80 & 59.72 & 50.74 & 58.20 & 54.77 & 49.45 \\
\quad + SAFA~\cite{hong2022safa}         & 65.98 & 60.43 & 51.20 & 59.80 & 55.71 & 49.89 \\
\quad + ConCutMix~\cite{pan2024enhanced} & 66.73 & 60.07 & 50.03 & 58.77 & 55.68 & 50.19 \\
\bottomrule
\end{tabular}
\end{table}

\section{Ablation Studies}
\textbf{Ablation for Class-wise Weight.}~We examine the influence of the class-wise weight $\Lambda$ on spectral regularization, as shown in Table~\ref{tab:lambda_ema_combined}. All models are trained under identical settings using ViT-Small on ImageNet-LT and CIFAR100-LT.
Incorporating $\Lambda$ consistently improves overall accuracy on both datasets, indicating that frequency-aware weighting mitigates the dominance of head categories and stabilizes optimization under imbalance. The absence of $\Lambda$ leads to skewed gradients toward frequent classes, which degrades tail performance and hinders convergence.
These observations confirm that $\Lambda$ effectively rebalances gradient contributions across categories, encouraging uniform learning dynamics and improving recognition of minority classes in long-tailed visual recognition.

% \begin{table}[t!]
% % 统一字号与间距（可按需调小/调大）
% \scriptsize
% \setlength{\tabcolsep}{2pt}
% \renewcommand{\arraystretch}{1.0}

% % 把两块内容装进一个恰好等于列宽的盒子里，防止越界或漂移
% \makebox[\linewidth][c]{%
% %
% \begin{minipage}[t]{0.48\linewidth}\vspace{0pt}\centering
% \captionof{table}{Effect of $\lambda_j$.}
% \label{tab:w_ablation}
% \begin{tabular}{lcc}
% \toprule
% \textbf{Datasets} & \textbf{w/ $\lambda_j$} & \textbf{w/o $\lambda_j$} \\
% \midrule
% ImageNet-LT  & 57.48 & 54.39 \\
% CIFAR100-LT  & 51.85 & 49.62 \\
% \bottomrule
% \end{tabular}
% \end{minipage}
% \hfill
% \begin{minipage}[t]{0.48\linewidth}\vspace{0pt}\centering
% \captionof{table}{Effect of EMA.}
% \label{tab:ema_ablation}
% \begin{tabular}{lcc}
% \toprule
% \textbf{Datasets} & \textbf{w/ EMA} & \textbf{w/o EMA} \\
% \midrule
% ImageNet-LT  & 57.48 & 55.77 \\
% CIFAR100-LT  & 51.85 & 50.20 \\
% \bottomrule
% \end{tabular}
% \end{minipage}
% } % end makebox

% \end{table}

\begin{figure*}[!htb]
  \centering
  \includegraphics[width=0.80\linewidth]{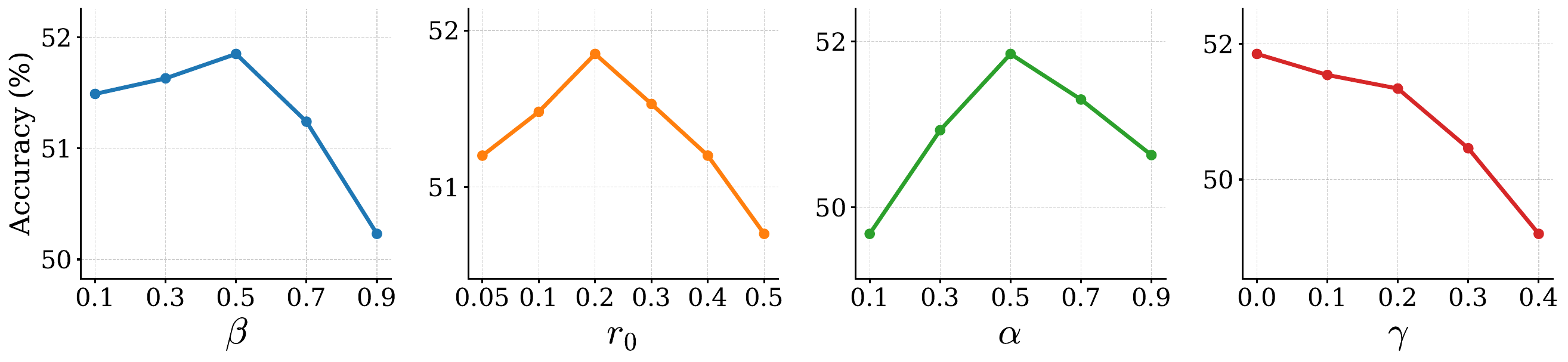}
    \caption{Ablation on four hyperparameters on CIFAR100-LT with ViT-Small (Top-1 accuracy).
    From left to right: EMA factor $\beta$, smoothing radius $r_{0}$, regularization weight $\alpha$, and margin gate $\gamma$.}
  \label{fig:hparam_curves}
\end{figure*}

\noindent\textbf{Ablation for EMA.}~We evaluate the effect of the EMA mechanism on stabilizing spectral estimation, as reported in Table~\ref{tab:lambda_ema_combined}. 
Introducing EMA leads to consistent improvements on both datasets, confirming that exponential averaging effectively smooths confusion updates and enhances training stability under imbalance. 
Without EMA, the estimation becomes highly variant across iterations, resulting in unstable convergence and reduced overall performance. 
These results demonstrate that EMA plays a vital role in maintaining robust spectral regularization. By maintaining smoother updates of the confusion matrix, it enables more robust convergence and contributes to overall performance gains in long-tailed recognition.

\noindent \textbf{Hyperparameter Analysis.}~We analyze the sensitivity of \texttt{CAR} to four key hyperparameters on CIFAR100-LT using ViT-Small, as shown in Figure~\ref{fig:hparam_curves}. From left to right, the plots correspond to the EMA factor $\beta$, smoothing radius $r_{0}$, regularization weight $\alpha$, and margin gate $\gamma$. Overall, \texttt{CAR} exhibits stable behavior across a wide range of configurations, suggesting that the model is not overly sensitive to precise parameter choices. A moderate EMA factor ($\beta=0.5$) provides the most stable moving average for spectral estimation, while $r_{0}=0.2$ achieves a good balance between suppressing noise and preserving discrimination. The regularization weight $\alpha$ reaches an optimum near $0.5$, and smaller $\gamma$ values yield better results by preventing boundary distortion under long-tailed conditions. In summary, the analysis highlights how different hyperparameters cooperatively shape the trade-off between stability and discriminative ability in \texttt{CAR}.

\begin{table}[htb]
\centering
\small
\setlength{\tabcolsep}{2.8pt}
\renewcommand{\arraystretch}{0.96}
\caption{Ablations for $\Lambda$ and EMA. Experiments are conducted on ImageNet-LT and CIFAR100-LT based on the ViT-Small.}
\label{tab:lambda_ema_combined}
\resizebox{0.75\columnwidth}{!}{%
\begin{tabular}{l l cc}
\toprule
\textbf{Factor} & \textbf{Setting} & \textbf{ImageNet-LT} & \textbf{CIFAR100-LT} \\
\midrule
\multirow{2}{*}{$\Lambda$ } 
  & w/o $\Lambda$   & 54.39 & 49.62 \\
   & w/ $\Lambda$    & 57.48 & 51.85 \\
\midrule
\multirow{2}{*}{EMA} 
  & w/o EMA         & 55.77 & 50.20 \\
    & w/ EMA          & 57.48 & 51.85 \\
\bottomrule
\end{tabular}
}
\end{table}

\section{Related Work}
\label{sec:relatedwork}
\noindent\textbf{Long-Tailed Learning.}~Long-tailed learning aims to address the severe class imbalance that commonly occurs in real-world datasets, where a few head classes dominate the training distribution while numerous tail classes have scarce samples~\cite{wang2019dynamic}. Such imbalance leads to biased decision boundaries and degraded generalization on rare categories. 
To mitigate this issue, a wide range of strategies have been developed~\cite{hong2021disentangling}. 
Re-sampling and re-weighting approaches (e.g., CB~\cite{cui2019classbalanced}, RS~\cite{shen2016relay}, RW~\cite{khan2017cost}, LDAM-DRW~\cite{cao2019learning}) rebalance the data distribution or modify loss weights to emphasize tail instances. 
Decoupled and representation-oriented learning methods (e.g., MiSLAS~\cite{zhong2021improving}, DisAlign~\cite{zhang2021distribution}, ResLT~\cite{cui2022reslt}, BALMS~\cite{ren2020balms}) refine feature spaces through balanced fine-tuning or class-specific calibration. 
Augmentation and meta-learning schemes such as Remix~\cite{chou2020remix}, MetaSAug~\cite{li2021metasaug}, and ConCutMix~\cite{pan2024enhanced} further enhance minority representation by generating more diverse visual patterns. 
Despite these advances, existing approaches mainly focus on adjusting sample frequency or loss weighting, while the confusion spectrum, a critical indicator of model bias, remains largely underexplored and motivates our confusional spectral regularization framework.

\noindent\textbf{Confusion Matrix-based Learning.}~The confusion matrix has long been a core diagnostic tool for analyzing model predictions, capturing both class-wise accuracy and inter-class misclassification patterns~\cite{morvant2012pac,machart2012confusion,kerrigan2021combining}.
Beyond evaluation, it has been employed to characterize classifier bias~\cite{li2023wat,wu2022confusion}, improve fairness~\cite{sun2024csr,wei2023cfa,jin2025enhancing}, and support robust optimization~\cite{yue2023revisiting}.
Calibration-based methods adjust posterior probabilities to reduce bias~\cite{lovell2022never}, while relation-based approaches build class-correlation graphs to regularize representations and encourage balanced decision boundaries~\cite{gortler2022neo,arias2023confusion}.
More recently, spectral analysis of the confusion matrix has revealed that its singular value spectrum reflects the dominance of major confusion directions~\cite{erbani2024confusion}.
Building on these insights~\cite{hasnain2020evaluating}, confusional spectral regularization constrains the spectral norm of the confusion matrix to suppress biased eigen-directions and enhance classifier fairness.
% Unlike prior studies that merely analyze or post-process confusion information, our method explicitly regularizes its spectrum throughout the training phase, achieving intrinsic bias mitigation and stable long-tailed generalization across diverse datasets.
%Unlike prior studies that utilize confusion statistics for analysis or auxiliary regularization, our method directly constrains the confusion matrix during training through spectral-norm minimization.

\section{Conclusion} \label{sec:conclusion}
In this work, we introduced a confusion-centric perspective for long-tailed recognition and established a new generalization upper bound that tightly connects the class-specific error to the spectral norm of a \emph{frequency-weighted confusion matrix}. This analysis reveals that controlling the spectral structure of inter-class confusions provides a principled route for improving generalization on the underrepresented categories. Guided by these insights, we proposed \texttt{CAR}, a \emph{Confusion-Aware Spectral Regularizer} that integrates a differentiable confusion matrix surrogate with an EMA-based estimator to enable stable and scalable optimization within standard training pipelines. Extensive experiments across CIFAR100-LT, ImageNet-LT, and iNaturalist demonstrate that \texttt{CAR} consistently improves both \emph{worst-class} and \emph{overall} accuracy, achieving state-of-the-art performance under both training-from-scratch and fine-tuning regimes. Moreover, \texttt{CAR} complements existing data-augmentation strategies such as ConCutMix, yielding further gains on both head and tail classes.

%We presented Confusion-Aware Spectral Regularizer (\texttt{CAR}), a principled framework for long-tailed recognition from a confusion-spectrum perspective. By linking worst-class error to the spectral norm of the frequency-weighted confusion matrix, \texttt{CAR} reduces class interference through spectral regularization and EMA-based estimation. Experiments across datasets, backbones, and imbalance factors show consistent gains in overall and tail performance. \texttt{CAR} functions as a general regularizer that complements existing augmentation methods, offering a simple and effective solution for fair and stable long-tailed visual recognition.

\section{Acknowledgment} \label{sec:acknowledgment}
The authors acknowledge the use of resources provided by the UKRI SLAIDER project, the MRC SLAIDER-QA project, the Isambard-AI National AI Research Resource (AIRR), and the Dutch national e-infrastructure, supported by the SURF Cooperative (Project EINF-17091). Isambard-AI is operated by the University of Bristol and funded by the UK Government’s Department for Science, Innovation and Technology (DSIT) via UK Research and Innovation and the Science and Technology Facilities Council [ST/AIRR/I-A-I/1023]. Finally, we thank the anonymous reviewers for their insightful comments, which significantly improved the quality of this paper.

{
    \small
    \bibliographystyle{ieeenat_fullname}
    \bibliography{main}
}
\clearpage
\setcounter{page}{1}
\maketitlesupplementary

% \section{Rationale}
% \label{sec:rationale}
% % 
% Having the supplementary compiled together with the main paper means that:
% % 
% \begin{itemize}
% \item The supplementary can back-reference sections of the main paper, for example, we can refer to \cref{sec:intro};
% \item The main paper can forward reference sub-sections within the supplementary explicitly (e.g. referring to a particular experiment); 
% \item When submitted to arXiv, the supplementary will already included at the end of the paper.
% \end{itemize}
% % 
% To split the supplementary pages from the main paper, you can use \href{https://support.apple.com/en-ca/guide/preview/prvw11793/mac#:~:text=Delete%20a%20page%20from%20a,or%20choose%20Edit%20%3E%20Delete).}{Preview (on macOS)}, \href{https://www.adobe.com/acrobat/how-to/delete-pages-from-pdf.html#:~:text=Choose%20%E2%80%9CTools%E2%80%9D%20%3E%20%E2%80%9COrganize,or%20pages%20from%20the%20file.}{Adobe Acrobat} (on all OSs), as well as \href{https://superuser.com/questions/517986/is-it-possible-to-delete-some-pages-of-a-pdf-document}{command line tools}.

\section{Proof for Proposition 3.2}
\label{appendix_proof}

We develop the proof based on the proof in \citet{jin2025enhancing,jin2020does,jin2025s} and the following PAC-Bayesian bound.
\begin{thm}[\citet{morvant2012pac}]
\label{thm:2.1}
Consider a training dataset $\mathcal{S}$ with $m$ samples drawn from a distribution $\mathcal{D}$ on $\mathcal{X} \times \mathcal{Y}$ with $\mathcal{Y} = \{1, \ldots, K\}$. 
Given a learning algorithm (e.g., a classifier) with prior and posterior distributions $P$ and $Q$ (i.e., $w+u$) on the weights respectively, 
for any $\delta > 0$, with probability $1-\delta$ over the draw of training data, we have that 
\begin{equation}
\| \mathbf{C}^Q_\mathcal{S} - \mathbf{C}^Q_\mathcal{D} \|_2 \leq \sqrt{\frac{8 K}{m_{min} - 8K} \left[ \mathrm{KL}(Q \| P) + \ln \left( \frac{m_{min}}{4\delta} \right) \right]},
\end{equation}
where $m_{min}$ represents the minimal number of examples from $\mathcal{S}$ which belong to the same class, $\mathbf{C}^Q_\mathcal{S} = \mathbb{E}_u\mathbf{C}^{f_{w+u}}_\mathcal{S}$, and $\mathbf{C}^Q_\mathcal{D} = \mathbb{E}_u\mathbf{C}^{f_{w+u}}_\mathcal{D}$.
\end{thm}
Let $\mathcal{S}_{u}$ denote the set of perturbations satisfying
\begin{equation}
\label{eq:su}
\mathcal{S}_{u} \subseteq\left\{u\Big|\max _{x \in \mathcal{X}}| f_{w+u}(x)-\left.f_{w}(x)\right|_{\infty}<\frac{\gamma}{4}
\right\}.
\end{equation}
Let $q$ be the probability density function of $u$.
We define a new distribution $\tilde Q$ restricted to $\mathcal{S}_{u}$, with density
\begin{equation}
\label{eq:qu}
\tilde{q}(\tilde{u})= \begin{cases} \frac{1}{z} q(\tilde{u}) & \tilde{u} \in \mathcal{S}_{u}, \\ 0 & \text { otherwise}, \end{cases}
\end{equation}
where $z$ is a normalizing constant.
By construction of $\tilde Q$, for all $\tilde u \sim \tilde Q$, we have
\begin{equation}
\label{eq:app1.1}
\begin{aligned}
\max _{x \in \mathcal{X}_B} | f_{w+\tilde u}(x)-\left.f_{w}(x)\right|_{\infty}<\frac{\gamma}{4}.
\end{aligned}
\end{equation}
For any $x \in \mathcal{X}$ such that $\arg\max_i f(x)[i] \ne y$, it follows that
\begin{equation}
f_{w+\tilde u}(x)[\arg\max_i f(x)[i]]+\frac{\gamma}{4} \ge f_{w+\tilde u}(x)[y]-\frac{\gamma}{4}.    
\end{equation}
Hence, for all $i \ne j$,
\begin{equation}
    (\mathbf{C}^{f}_\mathcal{D})_{ij} \le (\mathbf{C}^{\tilde Q}_{\mathcal{D},\frac{\gamma}{2}})_{ij}.
\end{equation}
According to the Perron–Frobenius theorem \citep{frobenius1912matrizen}, for all $1 \le i,j \le K$,
$\frac{\partial\|\mathbf{C}\|_2}{\partial (\mathbf{C})_{ij}} \ge 0$.
Therefore,
\begin{equation}
\|\mathbf{C}^{f}_\mathcal{D} \|_2 \le \|\mathbf{C}^{\tilde Q}_{\mathcal{D},\frac{\gamma}{2}}\|_2. 
\end{equation}
Using the inequality $|\|{A}\|_2-\|{B}\|_2| \le\|{A-B}\|_2$ and Theorem~\ref{thm:2.1}, we obtain
\begin{small}
\begin{equation}
\|\mathbf{C}^{\tilde Q}_{\mathcal{D},\frac{\gamma}{2}} \|_2 \le \|\mathbf{C}^{\tilde Q}_{\mathcal{S},\frac{\gamma}{2}} \|_2 + \sqrt{\frac{8 K}{m_{min} - 8K} \left[ \mathrm{KL}(\tilde Q \| P) + \ln \left( \frac{m_{min}}{4\delta} \right) \right]}.
\end{equation}
\end{small}Moreover, for any $x \in \mathcal{X}$, if there exists $\tilde u \in \tilde Q$ such that
$\max_{i \ne y} f_{w+\tilde u}(x)[i] + \tfrac{\gamma}{2} \ge f_{w+\tilde u}(x)[y]$,
then it must also hold that
\begin{equation}
\max_{i\ne y} f_{w}(x)[i]+\gamma \ge f_{w}(x)[y].    
\end{equation}
Thus for all $i\ne j$, we have
\begin{equation}
    (\mathbf{C}^{\tilde Q}_{\mathcal{S},\frac{\gamma}{2}})_{ij} \le (\mathbf{C}^{f_w}_{\mathcal{S},\gamma})_{ij}.
\end{equation}
Applying the same monotonicity argument from Perron–Frobenius, we obtain
\begin{equation}
\|\mathbf{C}^{\tilde Q}_{\mathcal{S},\frac{\gamma}{2}}\|_2 \le \|\mathbf{C}^{f}_{\mathcal{S},\gamma}\|_2.
\end{equation}
Next, let $\mathcal{S}_{u}^c$ denote the complement of $\mathcal{S}_{u}$ and $\tilde q^c$ the normalized density over $\mathcal{S}_{u}^c$.
From (\ref{eq:qu}), the KL divergence decomposes as
\begin{equation}
\mathrm{KL}(q\| p) = z\mathrm{KL}(\tilde q\| p) + (1-z)\mathrm{KL}(\tilde q^c\| p)-H(z),
\end{equation}
where $H(z)=-z \ln z-(1-z) \ln (1-z) \leq 1$ is the binary entropy function. 
Since $\mathrm{KL}$ is always positive, we get 
\begin{equation}
\begin{aligned}
\mathrm{KL}(\tilde{q} \| p)&=\frac{1}{z}[\mathrm{KL}(q \| p)+H(z))-(1-z) \mathrm{KL}(\tilde{q}^c \| p)]\\ 
&\leq 2(\mathrm{KL}(q \| p)+1).
\end{aligned}
\end{equation}
Thus we have 
\begin{equation}
2(\mathrm{KL}(w+u || P)+\ln \frac{3m_{min}}{4\delta})\ge \mathrm{KL}(w+\tilde u || P)+\ln \frac{m_{min}}{4\delta}.
\end{equation}

Therefore, combine the above equations, with probability at least $1-\delta$ over training dataset $\mathcal{S}$, we have
\begin{small}
\begin{equation}\nonumber
\begin{aligned}
&\|\mathbf{C}^{f}_\mathcal{D} \|_2 \le \|\mathbf{C}^{\tilde Q}_{\mathcal{D},\frac{\gamma}{2}}\|_2 \quad\quad\quad\quad\quad\quad  \\
&\le \|\mathbf{C}^{\tilde Q}_{\mathcal{S},\frac{\gamma}{2}} \|_2 + \sqrt{\frac{8 K}{m_{min} - 8K} \left[ \mathrm{KL}(\tilde Q \| P) + \ln \left( \frac{m_{min}}{4\delta} \right) \right]}  \\
&\le \|\mathbf{C}^{f}_{\mathcal{S},\gamma} \|_2 + \sqrt{\frac{8 K}{m_{min} - 8K} \left[ \mathrm{KL}(\tilde Q \| P) + \ln \left( \frac{m_{min}}{4\delta} \right) \right]}   \\
&\le \|\mathbf{C}^{f}_{\mathcal{S},\gamma} \|_2 + 4\sqrt{\frac{K}{m_{min} - 8K} \left[ \mathrm{KL}(Q \| P) + \ln \left( \frac{3m_{min}}{4\delta} \right) \right]}.  
\end{aligned}
\end{equation}
\end{small}

Following \citet{neyshabur2017pac}, the next proof proceeds in two main steps.
First, we determine the maximum perturbation $u$ that can be applied to the weights while still preserving the required margin $\gamma$.
Second, using this allowable perturbation, we evaluate the KL divergence term that appears in the PAC-Bayesian bound.
These two components together yield the desired generalization bound.

We consider a neural network with weight matrices ${W_l}$, $l\in\{1,...,n\}$, and normalize each matrix by its spectral norm $\|W_l\|_2$.
Let $\beta$ denote the geometric mean of the spectral norms:
$$\beta = \left(\prod_{l=1}^n \|W_l\|_2\right)^{\frac{1}{n}},$$
where $n$ is the number of weight matrices in the network. 
We construct a rescaled set of weights ${\widetilde{W}_l}$ by adjusting each $W_l$ according to
$$\widetilde{W}_l = \frac{\beta}{\|W_l\|_2} W_l.$$
Because the ReLU activation function is positively homogeneous, this reparameterization preserves the functional behavior of the network.
Thus, the resulting model $f_{\widetilde{W}}$ computes exactly the same output as the original model $f_{W}$, allowing us to work with the normalized parameterization without loss of generality.

Furthermore, note that the product of the spectral norms is preserved under this normalization: $$\prod_{l=1}^n \|W_l\|_2=\prod_{l=1}^n \|\widetilde{W}_l\|_2.$$ 
In addition, the ratio between the Frobenius norm and the spectral norm remains unchanged for every layer:
$$\frac{\|W_l\|_F}{\|W_l\|_2} = \frac{\|\widetilde{W}_l\|_F}{\|\widetilde{W}_l\|_2}.$$
Therefore, the excess error appearing in the theorem is invariant under this normalization.
It is thus sufficient to establish the result for the normalized weights $\widetilde{W}$, and we may assume without loss of generality that $\|W_l\|_2 = \beta$ for any layer $l$.

We take the prior distribution $P$ to be a zero-mean Gaussian with diagonal covariance $\sigma^2 \mathbf{I}$ and introduce perturbations $U \sim \mathcal{N}(0, \sigma^2 \mathbf{I})$, where $\sigma$ will later be chosen as a function of $\beta$.
Since the prior must not depend on the learned weights $W$ or their norm, $\sigma$ is selected using an estimate $\tilde{\beta}$ instead of $\beta$ itself.
To ensure coverage over all possible values of $\beta$, we compute the PAC-Bayesian bound for each $\tilde{\beta}$ in a predefined grid.
This gives a generalization guarantee for every $W$ satisfying: $|\beta - \tilde{\beta}| \leq \frac{1}{n} \beta$. 

Thus, each feasible $\beta$ is close to some $\tilde{\beta}$ in the grid, and applying a union bound across all grid values yields a uniform guarantee.
In particular, for any such pair $(\beta, \tilde{\beta})$, we have
$$\frac{1}{e} \beta^{n-1} \leq \tilde{\beta}^{n-1} \leq e \beta^{n-1}.$$

Following \citet{bandeira2021spectral}, and noting that each perturbation matrix $U_l$ satisfies $U_l \sim \mathcal{N}(0, \sigma^2 \mathbf{I})$ (equivalently, $u_l = \mathrm{vec}(U_l)$), we obtain the following tail bound on its spectral norm:
\begin{equation}
\label{eq:uwbound}
\mathbb{P}_{U_l \sim \mathcal{N}(0, \sigma^2 \mathbf{I})}\left[\|U_l\|_2 > t\right] \leq 2h \exp\left(-\frac{t^2}{2h\sigma^2}\right),
\end{equation}
where $h$ denotes the width of the hidden layers.
Applying a union bound across all $n$ layers, we conclude that with probability at least $1/2$, the perturbation $U_l$ in each layer is bounded by $\sigma \sqrt{2h \ln (4 nh)}$.

Plugging the bound from \citet{neyshabur2017pac}, we have that
\begin{equation}
\label{eq:cod1}
\begin{aligned}
&\max _{x \in \mathcal{X}}\left\|f_{w+U}(x)-f_{w}(x)\right\|_2 \leq e B \beta^n \sum_l \frac{\left\|U_l\right\|_2}{\beta} \\
& =e B \beta^{n-1} \sum_l\left\|U_l\right\|_2 \\
&\leq e^2 n B \tilde{\beta}^{n-1} \sigma \sqrt{2 h \ln (4 n h)} \leq \frac{\gamma}{4},
\end{aligned}
\end{equation}
where $B$ is the largest $\ell_2$ norm of input samples.

To ensure that (\ref{eq:cod1}) holds and using the fact that $\tilde{\beta}^{n-1} \le e\beta^{n-1}$, we choose the largest valid value of $\sigma$ as
\begin{equation}\nonumber
\sigma = \frac{\gamma}{114 n B \sqrt{h \ln (4 n h)}\prod_{l=1}^n \|W_l\|_2^{\frac{n-1}{n}}}.
\end{equation}

With this choice of $\sigma$, the perturbation $U$ satisfies the required margin condition.
We now compute the KL term for the selected prior $P$ and posterior $Q$:
\begin{equation}\nonumber
\begin{aligned}
&\mathrm{KL}(w+u \| P) \le \frac{\|w\|_2^2}{2\sigma^2}\\
&=\frac{\sum_{l=1}^n\|W_l\|_F^2}{2\sigma^2}\\
&\le \mathcal{O}\left(B^2n^2h\ln(nh)\frac{\prod_{l=1}^n \|W_l\|_2^2}{\gamma^2} \sum_{l=1}^n \frac{\|W_l\|^2_F}{\|W_l\|^2_2} \right).
\end{aligned}
\end{equation}
Then, we can give a union bound over different choices of $\tilde \beta$.
We only need to form the bound for $\left(\frac{\gamma}{2 B}\right)^{\frac{1}{n}} \leq \beta \leq\left(\frac{\gamma \sqrt{m}}{2 B}\right)^{\frac{1}{n}}$ which can be covered using a cover of size $nm^{\frac{1}{2n}}$ as discussed in \citet{neyshabur2017pac}.
Thus, with probability $\ge 1-\delta$, for any $\tilde \beta$ and for all $w$ such that $|\beta-\tilde{\beta}| \leq \frac{1}{n} \beta$, we have:
\begin{small}
\begin{equation}
\begin{aligned}
&\|\mathbf{C}^{f}_\mathcal{D} \|_2 \leq \| \mathbf{C}^{f}_{\mathcal{S},\gamma}\|_2 \\
&+ \mathcal{O}\left(\sqrt{\frac{K}{(m_{min} - 8K)\gamma^2} \left[ \Phi(f_w) + \ln \left( \frac{n m_{min}}{\delta} \right) \right]}\right),
\end{aligned}
\end{equation}
\end{small}where $\Phi(f_w)=B^2n^2h\ln(nh)\prod_{l=1}^n ||W_l||_2^2 \sum_{l=1}^n \frac{||W_l||^2_F}{||W_l||^2_2}$.

% \gaojie{Here!!}
To relate the spectral norm bound obtained above to adversarial performance, we use the known relationship between the $\ell_1$ norm and the spectral norm of a matrix. This allows us to convert the bound on $\|\mathbf{C}^{f}_{\mathcal{D}}\|_2$ into a bound on the $\ell_1$ norm, which directly characterizes the worst-class error.
In particular, for any confusion matrix $\mathbf{C} \in \mathbb{R}^{K \times K}$, we have
$\|\mathbf{C}^{f}_{\mathcal{D}} \|_1 \leq \nu' \|\mathbf{C}^{f}_{\mathcal{D}} \|_2$, where $\nu'$ is a constant that depends on the number of classes $K$ and is upper bounded by $\sqrt{K}$.

Then, given $\Lambda = \mathrm{diag}(\lambda_1, \dots, \lambda_K)$, for an adjusted constant $\nu$ corresponding to $\nu'$, for all $j$, we have
\begin{equation}\nonumber
\begin{aligned}
e_j &\leq \frac{1}{\lambda_{j}}\bigl\| \mathbf{C}_{\mathcal D}^{f} \Lambda \bigr\|_{1}
\;\\
&\le\;
\frac{\nu}{\lambda_{j}}\,\bigl\| \mathbf{C}_{\mathcal{S},\gamma}^{f} \Lambda \bigr\|_{2}\\
&+ \mathcal{O}\left(\sqrt{\frac{K}{(m_{min} - 8K)\gamma^2} \left[ \Phi(f_w) + \ln \left( \frac{n m_{min}}{\delta} \right) \right]}\right).
\end{aligned}
\end{equation}

Hence, proved. \hfill $\square$

\section{Stability Analysis of the EMA-based Confusion Estimator} \label{appedix:stable}

Let \(\mathbf{A}_t := \widehat{\mathbf{C}}_t\,\mathbf{W}\) and \(\mathcal{R}_t(f) = \|\mathbf{A}_t\|_2\).
Write the top singular triplet of \(\mathbf{A}_t\) as \((\sigma_t,\mathbf{u}_t,\mathbf{v}_t)\) with \(\|\mathbf{u}_t\|=\|\mathbf{v}_t\|=1\).
A standard subgradient of the spectral norm gives
\[
\frac{\partial \mathcal{R}_t}{\partial \mathbf{A}_t} \;=\; \mathbf{u}_t\,\mathbf{v}_t^\top,
\qquad
\frac{\partial \mathcal{R}_t}{\partial \widehat{\mathbf{C}}_t}
\;=\;
\mathbf{u}_t\,\mathbf{v}_t^\top \mathbf{W}^\top.
\]
Because \(\widehat{\mathbf{C}}_t=\beta\,\widehat{\mathbf{C}}_{t-1}+(1-\beta)\,\mathbf{C}_t\) and \(\widehat{\mathbf{C}}_{t-1}\) is treated as a constant (stop-gradient), the chain rule yields
\begin{align}
\frac{\partial \mathcal{R}_t}{\partial \mathbf{C}_t}
&=\;(1-\beta)\,\mathbf{u}_t\,\mathbf{v}_t^\top \mathbf{W}^\top,
\\[-2pt]
\nabla_w \mathcal{R}_t(f)
&=\;(1-\beta)\,\Big\langle \mathbf{u}_t\,\mathbf{v}_t^\top \mathbf{W}^\top,\;
\frac{\partial \mathbf{C}_t}{\partial w}\Big\rangle.
\end{align}
Hence the EMA introduces an explicit gain factor \(1-\beta\) on the gradient path from the batch confusion to the parameters:
\[
\big\|\nabla_w \mathcal{R}_t(f)\big\|
\;\le\;
(1-\beta)\,\|\mathbf{W}\|_2\,
\Big\|\frac{\partial \mathbf{C}_t}{\partial w}\Big\|_F,
\]
which attenuates stochastic spikes and improves step-size robustness.

\smallskip
\noindent\textbf{Variance reduction.}
Because \(\nabla_w \mathcal{R}_t\) depends on the current batch only through \(\mathbf{C}_t\), and \(\mathbf{u}_t\,\mathbf{v}_t^\top \mathbf{W}^\top\) evolves smoothly (above), the stochastic variance satisfies the proxy bound
\[
\mathrm{Var}\!\left[\nabla_w \mathcal{R}_t\right]
\;\lesssim\;
(1-\beta)^2\,\|\mathbf{W}\|_2^2\,
\mathrm{Var}\!\left[\tfrac{\partial \mathbf{C}_t}{\partial w}\right],
\]
showing EMA’s quadratic damping of gradient variance.

\section{More Experiments} \label{appedix:experiment}

\subsection{Experimental Setup} \label{appedix:experimentsetup}

\subsubsection{Datasets} \label{appedix:datasets}
We briefly introduce the four long-tailed benchmarks used in this study.
CIFAR100-LT is derived from the balanced CIFAR-100 dataset, containing 60,000 images of size 
32×32 from 100 object categories.
It serves as a compact benchmark for evaluating long-tailed classification under limited image resolution.
Tiny-ImageNet-LT is constructed from Tiny-ImageNet, which includes 200 classes with 500 training and 50 validation images per class in the balanced version.
It provides a mid-scale evaluation setting that bridges the gap between small-scale and large-scale datasets, featuring higher visual diversity and more complex backgrounds than CIFAR100-LT.
ImageNet-LT is a large-scale benchmark derived from ImageNet-2012, containing 1,000 categories with around 115K training images after sampling.
It retains the semantic richness and visual complexity of the full ImageNet while presenting a severe class imbalance, making it a standard testbed for long-tailed recognition at scale.
iNaturalist2018 is a real-world long-tailed dataset of 437K natural images from 8,142 species, collected from community-driven observations.
It exhibits extreme imbalance and fine-grained inter-class similarity, reflecting the natural frequency of species occurrences.
Following common practice, we evaluate performance on three subsets—Head, Medium, and Tail—based on the number of training samples per class.

\subsubsection{Implementation Details} \label{appedix:implementation_details}
All models are implemented in PyTorch and equipped with  NVIDIA RTX~3090 GPUs and A100 GPUs.
The following paragraphs describe the exact training procedures for both training from scratch and fine-tuning pre-trained models.
\paragraph{Training from Scratch.}
For long-tailed training from scratch, we follow standard practices commonly adopted in prior works. 
All models are optimized using AdamW with an initial learning rate of $1\times10^{-4}$, a weight decay follows~\cite{alshammari2022long}, and a batch size of 128. 
A cosine annealing scheduler is used to decay the learning rate from its initial value to zero throughout training.  
All Vision Transformer models are initialized unless otherwise noted.  

\paragraph{Fine-tuning Pre-trained Models.}
For experiments involving pre-trained backbones, we follow the commonly adopted fine-tuning protocol for ViT models.  
We fine-tune using AdamW with a learning rate of $2\times10^{-4}$ and batch size of 128.  
Models are fine-tuned for 100 epochs on CIFAR100-LT, ImageNet-LT, Tiny-ImageNet-LT, and iNaturalist2018.  
During fine-tuning, the patch embedding and transformer backbone weights are initialized from ImageNet-1K pre-trained checkpoints.

\subsection{Additional Worst-class Results} \label{appedix:worstresults}
In the Section~\ref{full:main}, we presented the core analysis of worst-class and overall accuracies on the ImageNet-LT benchmark, revealing two key gaps: 
(\textit{i}) the worst-class test accuracy lags significantly behind the overall accuracy, and 
(\textit{ii}) the worst-class test accuracy is much lower than its training counterpart.  
To further validate the generality of this observation, Figure~\ref{fig:worst_overall_bars_cifar} illustrates the corresponding results on CIFAR100-LT using ViT-Small as the backbone.
We observe consistent trends with those on ImageNet-LT. 
Conventional re-weighting or re-sampling methods achieve high worst-class accuracy on the training set but fail to generalize to unseen tail samples, leading to a pronounced drop in the test worst-class accuracy. 
Methods introducing feature-level balancing or meta-regularization slightly mitigate this issue but still show a large disparity between training and testing.  
In contrast, our proposed method achieves substantially higher test worst-class accuracy while maintaining competitive overall performance, demonstrating stable tail generalization and reduced overfitting.
These results confirm that the phenomena observed on large-scale ImageNet-LT are not dataset-specific but persist across different long-tailed settings.  
The improvement on CIFAR100-LT further verifies the robustness and transferability of our approach in alleviating the worst-class generalization gap.

\begin{figure*}[htb]
  \centering
  \captionsetup{aboveskip=2pt, belowskip=0pt} % 仅对本图生效
  \includegraphics[width=0.9\linewidth]{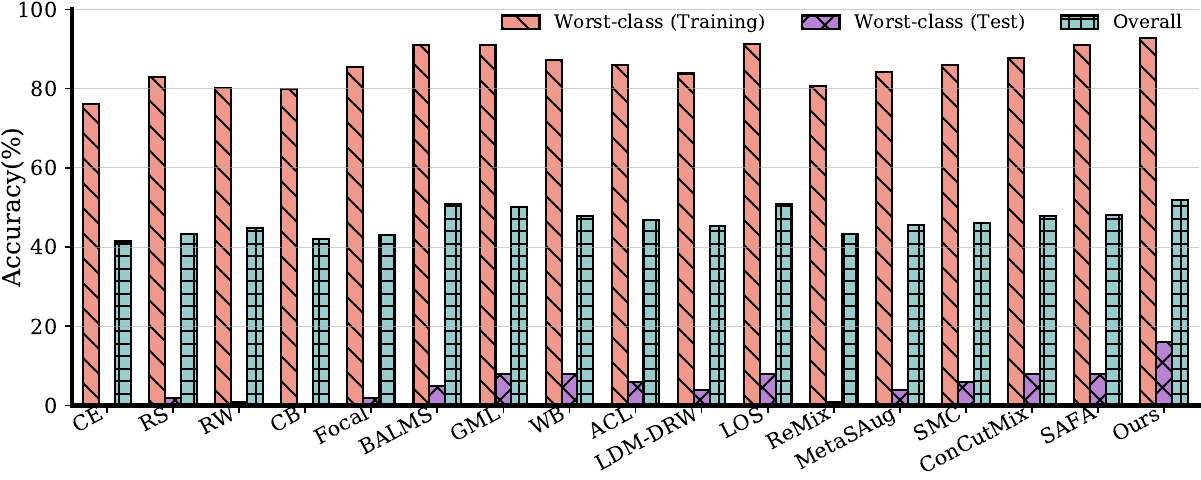}
  \vspace{-4pt}
  \caption{Comparison of worst-class and overall accuracies across representative long-tailed learning methods on CIFAR100-LT using Pretrained ViT-Small as the backbone. The three bars for each method correspond to the worst-class accuracy on the training set (left, red-hatched), the worst-class accuracy on the test set (middle, purple-hatched), and the overall test accuracy (right, green-hatched)}
  \label{fig:worst_overall_bars_cifar}
  \vspace{-8pt}
\end{figure*}

\subsection{Additional Results Across Backbones} \label{appedix:models}
\subsubsection{Additional Results Across Backbones on Pre-trained Models} 

To further demonstrate the backbone-agnostic property of our framework, as discussed in Section~\ref{full:backbone}, we conduct experiments using the pre-trained ResNet architecture on Tiny-ImageNet-LT under varying imbalance factors (IF = 1:50, 1:100, 1:200). 
The results, summarized in Table~\ref{tab:resnet_tiny_if}, confirm that our method maintains consistent superiority over existing long-tailed learning approaches across all imbalance settings.

As shown in the table, the proposed CAR achieves accuracies of 64.50\%, 59.49\%, and 54.16\% under increasing imbalance severity, clearly outperforming representative baselines such as BALMS, GML, and MetaSAug. 
These results validate that our method is not restricted to Transformer-based architectures but also generalizes effectively to convolutional networks. 
This cross-backbone consistency reinforces that our approach captures a universal optimization mechanism capable of mitigating imbalance-induced bias across diverse model families.

\begin{table}[htb]
\centering
\small
\setlength{\tabcolsep}{5pt}
\caption{Top-1 accuracy (\%) of pre-tranied ResNet on Tiny-ImageNet-LT under different imbalance factors (IF = 50, 100, and 200).
The best results are highlighted in \textbf{bold}.}
\label{tab:resnet_tiny_if}
\begin{tabular}{lccc}
\toprule
\textbf{Methods} & \textbf{IF=50} & \textbf{IF=100} & \textbf{IF=200} \\
\midrule
CE               & 54.86 & 50.41 & 47.90 \\
% RS~\cite{shen2016relay}                & 61.40 & 53.18 & 50.61 \\
% RW~\cite{khan2017cost}                & 62.01 & 54.69 & 50.58 \\
Focal~\cite{lin2017focal}       & 61.12 & 54.78 & 50.85 \\
CB~\cite{cui2019classbalanced}          & 61.49 & 53.01 & 49.41 \\
LDAM-DRW~\cite{cao2019learning}         & 57.79 & 52.18 & 49.42 \\
BALMS~\cite{ren2020balms}              & 63.16 & 58.70 & 52.95 \\
ReMix~\cite{chou2020remix}           & 56.25 & 52.75 & 49.06 \\
BBN~\cite{zhou2020bbn}              & 59.93 & 54.79 & 51.35 \\
MetaSAug~\cite{li2021metasaug}         & 61.77 & 54.74 & 50.77 \\
CMO~\cite{park2022majority}              & 60.74 & 55.85 & 50.05 \\
SAFA~\cite{hong2022safa}             & 62.54 & 57.53 & 52.55 \\
WB~\cite{alshammari2022long} & 60.03 & 54.01 & 50.29 \\
GML~\cite{du2023no}              & 62.95 & 58.11 & 53.12 \\
ConCutMix~\cite{pan2024enhanced}        & 60.36 & 56.19 & 51.48 \\
LOS~\cite{sunrethinking}             &63.04   &58.23	   &53.27 \\

\rowcolor{red!10}
\texttt{CAR}  (Ours)            & \textbf{64.50} & \textbf{59.49} & \textbf{54.16} \\
% \quad + ReMix~\cite{chou2020remix}      & 66.88 & 60.85 & 55.60 \\
% \quad + MetaSAug~\cite{li2021metasaug}   & 66.10 & 61.41 & 55.73 \\
% \quad + CMO~\cite{park2022majority}        & 67.84 & 61.06 & 56.51 \\
% \quad + ConCutMix~\cite{pan2024enhanced}  & 67.99 & 62.97 & 56.35 \\
% \quad + + SAFA~\cite{hong2022safa}        & 68.14 & 62.73 & 57.47 \\
\bottomrule
\end{tabular}
\end{table}

\subsubsection{Additional Results across Different ViT Backbone Sizes} 

To further evaluate the scalability of our framework, Table~\ref{tab:tiny_backbones_appendix} reports the performance on Tiny-ImageNet-LT using three pre-trained ViT backbones of different model sizes (ViT-Tiny, ViT-Base, and ViT-Large) with an imbalance factor of 1:100. 
This experiment complements the backbone generalization analysis presented in Section~\ref{full:backbone} and aims to assess whether the proposed method consistently benefits larger-capacity models.

Across all backbone sizes, our method achieves the highest top-1 accuracy, improving upon strong baselines such as BALMS, GML, and SAFA. 
As the model capacity increases from ViT-Tiny to ViT-Large, all methods exhibit overall performance gains due to stronger feature expressiveness. 
However, the relative improvement brought by our approach remains stable, showing improvement gain over the best competing method under each configuration. 
This indicates that our regularization mechanism continues to enhance class balance and generalization regardless of network scale.

The results here align with the findings on ImageNet-LT and CIFAR100-LT in Section~\ref{full:backbone}: our framework scales effectively with model size. 
This cross-scale consistency confirms that the proposed approach is not only architecture-agnostic but also size-agnostic, providing reliable benefits for both lightweight and large pre-trained transformer backbones.

\begin{table}[htb]
\centering
\footnotesize
\setlength{\tabcolsep}{3.6pt} 
\caption{Top-1 accuracy (\%) on Tiny-ImageNet-LT with different Pre-trianed ViT backbone sizes on Tiny-ImageNet-LT with an IF of 1:100.
The best results are highlighted in \textbf{bold}.}
\label{tab:tiny_backbones_appendix}
\begin{tabular}{lccc}
\toprule
\textbf{Methods} & \textbf{ViT-Tiny} & \textbf{ViT-Base} & \textbf{ViT-Large} \\
\midrule
CE                        & 40.64 & 65.19 & 70.40 \\
% RS~\cite{shen2016relay}                        & 44.34 & 69.81 & 75.19 \\
% RW~\cite{khan2017cost}                        & 45.46 & 69.30 & 76.82 \\
Focal~\cite{lin2017focal}                 & 45.66 & 69.64 & 77.60 \\
CB~\cite{cui2019classbalanced}            & 44.08 & 68.36 & 75.05 \\
LDAM-DRW~\cite{cao2019learning}           & 42.28 & 67.61 & 73.59 \\
BALMS~\cite{ren2020balms}                 & 48.18 & 72.21 & 79.88 \\
ReMix~\cite{chou2020remix}                & 43.28 & 67.91 & 74.48 \\
BBN~\cite{zhou2020bbn}                  & 45.57 & 69.72 & 76.03 \\
MetaSAug~\cite{li2021metasaug}            & 46.74 & 69.62 & 75.30 \\
CMO~\cite{park2022majority}            & 45.62 & 70.53 & 75.25 \\
SAFA~\cite{hong2022safa}                  & 46.56 & 71.02 & 78.16 \\
WB~\cite{alshammari2022long}              & 44.09 & 68.50 & 75.68 \\
GML~\cite{du2023no}                       & 48.43 & 71.99 & 79.38 \\
ConCutMix~\cite{pan2024enhanced}          & 46.38 & 71.38 & 77.51 \\
LOS~\cite{sunrethinking}                   &  47.55	&72.86	&80.47\\
\rowcolor{red!10}
\texttt{CAR} (Ours)                            & \textbf{49.44} & \textbf{74.83} & \textbf{81.65} \\
% \quad + ReMix~\cite{chou2020remix}        & 51.69 & 75.66 & 82.56 \\
% \quad + MetaSAug~\cite{li2021metasaug}    & 51.82 & 75.65 & 83.70 \\
% \quad + CMO~\cite{park2022majority}    & 51.36 & 77.71 & 83.89 \\
% \quad + ConCutMix~\cite{pan2024enhanced}  & 53.60 & 76.09 & 83.13 \\
% \quad + SAFA~\cite{hong2022safa}          & 53.45 & 75.38 & 82.19 \\
\bottomrule
\end{tabular}
\end{table}

% \subsection{Effect of Gradient Detachment}

% To examine the effect of applying the stop-gradient operation in the margin term of the differentiable confusion surrogate, we conduct an ablation study on CIFAR100-LT using ViT-Small. As shown in Table~\ref{tab:detach_ablation}, incorporating stop-gradient yields a clear improvement of +0.94\% in top-1 accuracy. This indicates that detaching the gradient flow from the margin computation helps stabilize optimization and prevents mutual interference between the classification logits and the confusion-based regularization. Without the stop-gradient operation, the margin term tends to overreact to local fluctuations in the class predictions, leading to suboptimal convergence. Therefore, applying stop-gradient proves essential for maintaining a smoother optimization landscape and achieving better overall generalization.

% \begin{table}[htb]
%   \centering
%   \captionof{table}{Effect of applying stop-gradient on the margin term in the differentiable confusion surrogate.}
%   \label{tab:detach_ablation}
%   \begin{tabular}{lc}
%     \toprule
%     \textbf{Variant} & \textbf{Accuracy (\%)} \\
%     \midrule
%     w/ stop-gradient in margin term & 51.85 \\
%     w/o stop-gradient               & 50.91 \\
%     \bottomrule
%   \end{tabular}
% \end{table}

\subsection{Additional Results Across Imbalance Factors based on Pre-trained Models} \label{appedix:if}

To further verify the generalization consistency of our framework, we conduct additional experiments using the pre-trained ViT-Small backbone on CIFAR100-LT and Tiny-ImageNet-LT under different imbalance factors (IF = 50 and 200). 
As reported in Table~\ref{tab:lt_ratio_appendix}, our method consistently achieves the best performance across both datasets, confirming its strong adaptability to diverse long-tailed scenarios.

Specifically, on Tiny-ImageNet-LT, our approach attains the highest accuracy of 75.87\% and 65.79\% under IF=50 and IF=200, respectively, surpassing existing advanced baselines such as BALMS, GML, and SAFA. 
A similar performance trend is observed on CIFAR100-LT, where our model maintains leading results of 81.11\% (IF=50) and 73.11\% (IF=200), demonstrating stable improvements over all representative competitors. 
Compared to conventional methods (e.g., ReMix, MetaSAug, and CMO), which exhibit substantial performance degradation as imbalance severity increases, our framework effectively preserves stable generalization.

Overall, these results confirm that our \texttt{CAR} not only mitigates overfitting to dominant classes but also enables robust transferability from pre-trained representations to imbalanced downstream datasets, reinforcing the conclusions drawn in Section~\ref{full:if}.

\begin{table}[htb]
\centering
\begingroup
\small
\setlength{\tabcolsep}{4pt}
\caption{Top-1 accuracy (\%) under different imbalance factors (IF) on Tiny-ImageNet-LT and CIFAR100-LT using pre-trained ViT-Small as the backbone. 
The best results are highlighted in \textbf{bold}.}
\label{tab:lt_ratio_appendix}
\begin{tabular}{lcccc}
\toprule
\multirow{2}{*}{\textbf{Methods}} & \multicolumn{2}{c}{\textbf{Tiny-ImageNet-LT}} &
\multicolumn{2}{c}{\textbf{CIFAR100-LT}} \\
\cmidrule(lr){2-3}\cmidrule(lr){4-5}
& IF=50 & IF=200 & IF=50 & IF=200 \\
\midrule
CE                 & 65.12 & 54.08 & 72.85 & 65.64 \\
% RS~\cite{shen2016relay}       & 69.55 & 57.23 & 76.71 & 67.38 \\
% RW~\cite{khan2017cost}        & 70.13 & 59.92 & 77.70 & 68.39 \\
Focal~\cite{lin2017focal}      & 70.24 & 59.56 & 78.39 & 69.51 \\
CB~\cite{cui2019classbalanced} & 69.45 & 57.62 & 76.81 & 67.65 \\
LDAM-DRW~\cite{cao2019learning} & 68.43 & 58.58 & 74.68 & 66.44 \\
BALMS~\cite{ren2020balms}      & 73.39 & 63.52 & 80.40 & 71.84 \\
ReMix~\cite{chou2020remix}     & 68.51 & 58.36 & 74.69 & 66.85 \\
BBN~\cite{zhou2020bbn}      & 70.31 & 60.82 & 77.65 & 68.75 \\
MetaSAug~\cite{li2021metasaug} & 70.32 & 61.65 & 77.12 & 69.42 \\
CMO~\cite{park2022majority} & 71.25 & 62.13 & 78.10 & 69.07 \\
SAFA~\cite{hong2022safa}       & 73.12 & 63.10 & 80.06 & 71.15 \\
WB~\cite{alshammari2022long}   & 70.78 & 61.44 & 78.23 & 69.79 \\
GML~\cite{du2023no}            & 73.86 & 63.94 & 80.26 & 72.13 \\
ConCutMix~\cite{pan2024enhanced} & 72.64 & 62.93 & 79.47 & 70.97 \\
LOS~\cite{sunrethinking}       & 74.23 & 64.10 & 80.55 & 72.54 \\
\rowcolor{red!10}
\texttt{CAR} (Ours)  & \textbf{75.87} & \textbf{65.79} & \textbf{81.11} & \textbf{73.11} \\
\bottomrule
\end{tabular}
\endgroup
\end{table}

\subsection{Additional Results of Complementary to Data Augmentation}  \label{appedix:dataaug}

\subsubsection{Additional Results of Complementary to Data Augmentation on Training from Scratch}
To evaluate the compatibility of our framework with long-tailed data augmentation techniques under training from scratch settings, as presented in Section~\ref{full:dataaug} we combine our method with several representative strategies, including ReMix~\cite{chou2020remix}, MetaSAug~\cite{li2021metasaug}, CMO~\cite{park2022majority}, ConCutMix~\cite{pan2024enhanced}, and SAFA~\cite{hong2022safa}. 
As shown in Table~\ref{tab:backbone_2_appendix}, our method consistently improves the performance of all augmentation baselines across diverse backbones, including ViT variants (Tiny, Base, Large), ResNet, and Swin.

Specifically, the integration with ConCutMix and SAFA achieves the most significant gains (e.g., 69.26\% $\rightarrow$ 73.39\% on ViT-Large and 56.38\% $\rightarrow$ 60.62\% on Swin), while other combinations such as ReMix and CMO also benefit from 1–2\% accuracy improvements. 
These consistent gains highlight that our frequency-weighted spectral regularization effectively complements data-level balancing strategies by further stabilizing optimization and reducing class-wise dominance.

Overall, these results confirm that the proposed framework not only generalizes across different model architectures but also synergizes with diverse long-tailed augmentation methods, delivering stable and transferable improvements even when trained from scratch.

\begin{table}[htb]
\centering
\footnotesize
\setlength{\tabcolsep}{3pt}
\renewcommand{\arraystretch}{0.96}
\caption{Top-1 accuracy (\%) on ImageNet-LT across different backbones. 
Results include ViT variants (Tiny/Base/Large), ResNet, and Swin.}
\label{tab:backbone_2_appendix}
\begin{tabular}{lccccc}
\toprule
\textbf{Methods} & \textbf{ViT-Tiny} & \textbf{ViT-Base} & \textbf{ViT-Large} & \textbf{ResNet} & \textbf{Swin} \\
\midrule
\texttt{CAR}   (Ours)                            & 46.35 & 63.79 & 69.26 & 50.27 & 56.38 \\
\quad + ReMix~\cite{chou2020remix}        & 47.77 & 65.55 & 71.02 & 51.93 & 58.05 \\
\quad + MetaSAug~\cite{li2021metasaug}    & 48.62 & 66.22 & 71.34 & 52.49 & 58.61 \\
\quad + CMO~\cite{park2022majority}    & 48.88 & 67.52 & 72.53 & 53.66 & 59.13 \\
\quad + SAFA~\cite{hong2022safa}          & 50.94 & 68.77 & 74.90 & 54.34 & 60.62 \\
\quad + ConCutMix~\cite{pan2024enhanced}  & 50.11 & 67.40 & 73.39 & 53.86 & 60.00 \\

\bottomrule
\end{tabular}
\end{table}

\subsubsection{Additional Results of Complementary to Data Augmentation on Fine-tuning from Pre-trained Models}

To further examine the general applicability of our framework under pre-trained settings, as given in Section~\ref{full:dataaug}, we integrate it with representative long-tailed data augmentation approaches, including ReMix~\cite{chou2020remix}, MetaSAug~\cite{li2021metasaug}, CMO~\cite{park2022majority}, ConCutMix~\cite{pan2024enhanced}, and SAFA~\cite{hong2022safa}. 
As shown in Table~\ref{tab:ours_aug_appendix}, our method consistently improves these augmentation baselines across diverse pre-trained backbones, including ViT variants (Tiny, Base, Large), ResNet, and Swin.

Specifically, when combined with ConCutMix and SAFA, our approach achieves notable accuracy gains (e.g., 81.65\% $\rightarrow$ 84.13\% on ViT-Large and 71.33\% $\rightarrow$ 74.09\% on Swin), while also improving ReMix and CMO by around 2\% under the same backbone. 
These results demonstrate that the proposed frequency-weighted spectral regularization remains effective even with pre-trained feature representations, providing complementary benefits to advanced data augmentation pipelines.

Overall, this synergy between our regularization framework and long-tailed data augmentation methods further validates the adaptability of our methos, ensuring stable performance improvements across both Transformer and convolutional backbones under pre-trained initialization.

\begin{table}[htb]
\centering
\begingroup
\footnotesize
\setlength{\tabcolsep}{3pt}
\renewcommand{\arraystretch}{0.96}
\caption{Top-1 accuracy (\%) on Tiny-ImageNet-LT across different pre-trained backbones. 
Results include ViT variants (Tiny/Base/Large), ResNet, and Swin.}
\label{tab:ours_aug_appendix}
\begin{tabular}{lccccc}
\toprule
\textbf{Methods} & \textbf{ViT-Tiny} & \textbf{ViT-Base} & \textbf{ViT-Large} & \textbf{ResNet} & \textbf{Swin} \\
\midrule
\texttt{CAR}  (Ours)                      & 49.44 & 74.83 & 81.65 & 59.49 & 71.13 \\
\quad + ReMix~\cite{chou2020remix}       & 51.69 & 75.66 & 82.56 & 60.85 & 72.75 \\
\quad + MetaSAug~\cite{li2021metasaug}   & 51.82 & 75.65 & 83.70 & 61.41 & 73.53 \\
\quad + CMO~\cite{park2022majority}   & 51.36 & 76.71 & 83.89 & 61.06 & 73.06 \\
\quad + SAFA~\cite{hong2022safa}         & 53.45 & 77.38 & 84.19 & 62.73 & 74.09 \\
\quad + ConCutMix~\cite{pan2024enhanced} & 53.60 & 77.09 & 84.13 & 62.97 & 73.47 \\

\bottomrule
\end{tabular}
\endgroup
\end{table}

\subsubsection{Additional Results of Complementary to Data Augmentation under Different Imbalanced Factors}

We further evaluate the compatibility of our framework with long-tailed data augmentation strategies under pre-trained settings with different imbalanced factors (IF) to confirm the conclusion in Section~\ref{full:dataaug}. 
As shown in Table~\ref{tab:lt_aug_v2_appendix}, we combine the proposed method with representative augmentation approaches, including ReMix~\cite{chou2020remix}, MetaSAug~\cite{li2021metasaug}, CMO~\cite{park2022majority}, ConCutMix~\cite{pan2024enhanced}, andSAFA~\cite{hong2022safa}, and evaluate them on Tiny-ImageNet-LT and CIFAR100-LT using a pre-trained ViT-Small backbone under varying imbalance factors (IF = 50, 100, 200).

Across all imbalance settings, our method consistently improves the baselines, demonstrating strong compatibility and stability. 
For instance, on Tiny-ImageNet-LT, the integration withConCutMix and SAFA achieves notable gains, while on CIFAR100-LT, the combined models further enhance accuracy under moderate imbalance (IF=100). 
These consistent improvements suggest that our method effectively complements data-level augmentation methods by mitigating class imbalance and enhancing feature generalization.

Overall, these results demonstrate that our method, when initialized from pre-trained weights, consistently integrates with diverse long-tailed data-augmentation strategies across different datasets and imbalance factors, yielding further performance gains.

% \begin{table}[htb]
% \centering
% \footnotesize
% \setlength{\tabcolsep}{2.8pt}
% \renewcommand{\arraystretch}{0.96}
% \caption{Top-1 Accuracy (\%) of Pre-trained ViT-Small on Tiny-ImageNet-LT and CIFAR100-LT under different imbalance factors (IF). 
% ``+'' indicates the combination of our method with other long-tailed data augmentation methods.}
% \label{tab:lt_aug_v2_appendix}
% \begin{tabular}{l|ccc|ccc}
% \toprule
% \multirow{2}{*}{Methods} &
% \multicolumn{3}{c|}{Tiny-ImageNet-LT} &
% \multicolumn{3}{c}{CIFAR100-LT} \\
% & IF=50 & IF=100 & IF=200 & IF=50 & IF=100 & IF=200 \\
% \midrule
% \texttt{CAR} (Ours)         & 75.87 & 72.97 & 65.79 & 81.11  & 79.37 & 73.11 \\
% \quad + ReMix~\cite{chou2020remix}    & 76.53 & 73.62 & 66.50 & 82.03 & 80.37 & 74.34 \\
% \quad + MetaSAug~\cite{li2021metasaug} & 76.96 & 74.07 & 67.35 & 82.11 & 81.07 & 74.31 \\
% \quad + CMO~\cite{park2022majority}    & 77.52 & 74.78 & 67.79 & 82.23 & 81.65 & 75.06 \\
% \quad + ConCutMix~\cite{pan2024enhanced} & 77.89 & 75.84 & 68.90 & 83.03 & 82.12 & 76.73 \\
% \quad + SAFA~\cite{hong2022safa}   & 78.32 & 75.65 & 68.36 & 83.18 & 82.61 & 77.17 \\
% \bottomrule
% \end{tabular}
% \end{table}

\begin{table}[htb]
\centering
\footnotesize
\setlength{\tabcolsep}{2.8pt}
\renewcommand{\arraystretch}{0.96}
\caption{Top-1 accuracy (\%) of pre-trained ViT-Small on Tiny-ImageNet-LT and CIFAR100-LT under different imbalance factors (IF).
``+'' indicates the combination of our method with other long-tailed data augmentation methods.}
\label{tab:lt_aug_v2_appendix}
\begin{tabular}{l V ccc V ccc}
\toprule
\multirow{2}{*}{\textbf{Methods}} &
\multicolumn{3}{c}{\textbf{Tiny-ImageNet-LT}} &
\multicolumn{3}{c}{\textbf{CIFAR100-LT}} \\
\cmidrule(lr){2-4}\cmidrule(lr){5-7}
& IF=50 & IF=100 & IF=200 & IF=50 & IF=100 & IF=200 \\
\midrule
\texttt{CAR} (Ours)                       & 75.87 & 72.97 & 65.79 & 81.11 & 79.37 & 73.11 \\
\quad + ReMix~\cite{chou2020remix}       & 76.53 & 73.62 & 66.50 & 82.03 & 80.37 & 74.34 \\
\quad + MetaSAug~\cite{li2021metasaug}   & 76.96 & 74.07 & 67.35 & 82.11 & 81.07 & 74.31 \\
\quad + CMO~\cite{park2022majority}      & 77.52 & 74.78 & 67.79 & 82.23 & 81.65 & 75.06 \\
\quad + SAFA~\cite{hong2022safa}         & 78.32 & 75.65 & 68.36 & 83.18 & 82.61 & 77.17 \\
\quad + ConCutMix~\cite{pan2024enhanced} & 77.89 & 75.84 & 68.90 & 83.03 & 82.12 & 76.73 \\
\bottomrule
\end{tabular}
\end{table}

\end{document}